\documentclass[journal,twoside,web]{ieeecolor}

\usepackage{generic}
\usepackage{cite}
\usepackage{amsmath,amssymb,amsfonts}
\usepackage{graphicx}
\usepackage{algorithm,algorithmic}
\usepackage{hyperref} 
\usepackage{booktabs}
\usepackage{textcomp}
\usepackage{makecell}
\usepackage{bm}
\usepackage{caption}
\def\BibTeX{{\rm B\kern-.05em{\sc i\kern-.025em b}\kern-.08em
    T\kern-.1667em\lower.7ex\hbox{E}\kern-.125emX}}
\allowdisplaybreaks[4]

\graphicspath{{picturesNew/}}
\newtheorem{definition}{Definition}
\newtheorem{assump}{\bf{Assumption}}
\newtheorem{theorem}{Theorem}
\newtheorem{lemma}{Lemma}
\newtheorem{prop}{Proposition}
\newtheorem{rmk}{Remark}

\begin{document}
\title{Neural Operator Feedback for a First-Order PIDE with Spatially-Varying State Delay}
\author{Jie Qi, Jiaqi Hu, Jing Zhang, and Miroslav Krstic
\thanks{The first three authors supported by the National Natural Science Foundation of China (62173084, 62403305), the Project of Science and Technology Commission of Shanghai Municipality, China (23ZR1401800).}
\thanks{Jie Qi and Jiaqi Hu are with the College of Information Science and Technology, Shanghai 201620, China (e-mail: jieqi@dhu.edu.cn, jiaqihu@mail.dhu.edu.cn). }
\thanks{Jing Zhang is with the College of Information Engineering, Shanghai Maritime University, Shanghai 200135, China (e-mail: zhang.jing@shmtu.edu.cn).}
\thanks{Miroslav Krstic is with the Department of Mechanical Aerospace Engineering, University of California, San Diego, CA 92093 USA (e-mail: krstic@ucsd.edu).
}
}

\maketitle
 \bstctlcite{BSTcontrol}
\begin{abstract}
A transport PDE with a spatial integral and  recirculation with constant delay has been a benchmark for neural operator approximations of PDE backstepping controllers. Introducing a spatially-varying delay into the model gives rise to a gain operator defined through integral equations which the operator's input---the varying delay function---enters in previously unencountered manners, including in the limits of integration and as the inverse of the `delayED time' function. This, in turn, introduces novel mathematical challenges in estimating the operator's Lipschitz constant. The backstepping kernel function having two branches endows the feedback law with a two-branch structure, where only one of the two feedback branches depends on both of the kernel branches. For this rich feedback structure, we propose a neural operator approximation of such a two-branch feedback law and prove the approximator to be semiglobally practically stabilizing. With numerical results we illustrate the training of the neural operator and its stabilizing capability. 
\end{abstract}

\begin{IEEEkeywords}
First-order hyperbolic PIDE, PDE backstepping, DeepONet, Spatially-varying delay, Learning-based control. 
\end{IEEEkeywords}

\section{Introduction}
\label{sec:introduction}
\IEEEPARstart{T}{his} paper considers the  first-order partial integro-differential equation (PIDE)   system
\begin{align}\label{eq:main-x0}
 	\partial_t x(s,t)&= -\partial_s x (s,t)+ \int^1_{s} \! f(s,q)x(q,t)dq
  \nonumber\\ &+c(s)x(1,t-\tau(s)), 
  ~s\in (0, 1),~t>0,
  \\ \label{eq:bnd-x}
 	x(0,t)&=  U(t), \\
 	x(s,0)&=x_0(s),\\
 	x(s,h)&=0, \qquad h\in[-\bar \tau,0),       \label{eq:x-before}
   \end{align}
 where $ \partial_s =\frac{\partial }{\partial s},~\partial_t =\frac{\partial }{\partial t}$, $\bar{\tau} = \sup_{s \in [0,1]} \tau(s)$, and extends the result from \cite{zhang2021compensation,qi2024neural}. Paper \cite{zhang2021compensation} solved the problem by the PDE backstepping method, and paper \cite{qi2024neural} propose neural operators to learn the  control gain functions (kernel functions) and the observer  gain functions for systems with a constant delay.  
The extension is both challenging and meaningful.
  
First, we employ a single DeepONet directly to learn the backstepping control operator, which consists of two distinct branches depending on the delay profile, each containing kernel functions with their own piecewise definitions. 
 This unified neural operator scheme completes `once and for all' \cite{bhan2024neural} once trained, the neural operator can quickly
 generate controllers for new delay function without recalculating kernels or complex integrals in the feedback, enhancing real time applicability. Importantly, we theoretically establish that the delay-dependent, two-branch backstepping control operator can be approximated by a single neural operator by proving its Lipschitz continuity. 
 
 \begin{figure}[htbp]    
 	\centering
 	\includegraphics[width=0.45\textwidth]{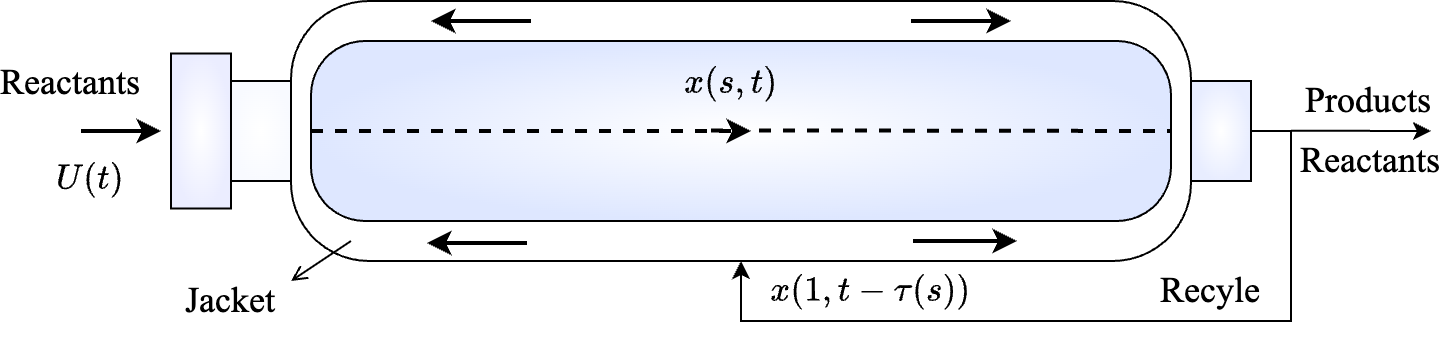}
 	\caption{The sketch of plug-flow tubular reactor with recycle.}
    \label{fig:plug-flow}
 \end{figure}

Spatially-varying delays is specially arise in physical systems where transport dynamics depend on position. For instance, in tokamak fusion devices for plasma temperature regulation, the control input modulates electron heating via neutral-beam injectors and RF antennas, induce delays that vary along magnetic field lines due to position-dependent transport speeds \cite{mavkov2017distributed,mameche2019nonlinear}. 
Another example of spatially-varying delay occurs in recycled tubular reactors \cite{reilly1966dynamics} (Fig. \ref{fig:plug-flow}). In this setup, recycled heat returned to the jacket for counter-current exchange causes transport delays that vary with position. Although many studies address delayed PDE control \cite{selivanov2018improved,krstic2009control,katz2020constructive}, research on spatially-varying delayed PDEs is limited due to difficulty of compensating different delays across space. Robust predictor feedback for parabolic PDEs, assuming small delay deviations from nominal values, has been developed in \cite{lhachemi2021robustness}, while backstepping methods are applied to account for delays without a nominal setting \cite{qi2020compensation,guan2022radially,zhang2021compensation}.

Nevertheless, controllers for spatially-varying delays, such as those using the backstepping method, involve complex structures with state and historical delayed variable feedback in form of piecewise integration and require to solve intricate kernel functions. Any change in the delay profile necessitates recomputation of the controller, leading to high computational costs and limited real-time applicability. 

In this context, neural operators\cite{lu2021learning,li2023fourier}, particularly DeepONet, offer a compelling alternative. DeepONet learns mappings between function spaces and generates PDE solutions in real time once trained. Its adaptability to new input functions and its theoretical guarantee for approximating continuous operators with arbitrary accuracy \cite{lanthaler2022error} make it highly effective for real-time control of PDEs \cite{wang2025deep,lamarque2024gain,zhang2023neural,lamarque2025gain}. 
Starting from \cite{bhan2024neural}, a transport PDE with recirculation and a spatial integral has proven a valuable benchmark for nonlinear operator approximations of backstepping. Our extension \cite{qi2024neural} with a constant delay in recirculation introduced a two-branch structure in the operator analysis. Krstic et al. \cite{krstic2024neural} extend the method to parabolic PDEs. Further studies have applied DeepONet to accelerate kernel equation computations in adaptive control frameworks \cite{bhan2023operator, lamarque2024adaptive, bhan2024adaptive}  enhancing the real-time performance of delay-compensated PDE controllers. Additionally, Lee et al. \cite{lee2024hamilton} integrated DeepONet with physical information to solve Hamilton-Jacobi-Bellman equations for optimal control.

 In this paper, with a spatially varying delay, we raise significantly the mathematical challenges in estimating the operator’s Lipschitz constant, as the gain operator defined through integral equations has the operator’s input (the spatially varying delay function) entering in previously unencountered ways, including within the integration limits and as the inverse of a ``delayed time" function. Additionally, the two-branched kernel function endows the feedback law also with a two-branch structure, with only one branch depending on both of the kernel branches. 
 To address these difficulties, we train a single DeepONet to approximate this rich feedback structure, offering implementation simplicity and the ability to automatically realize branch-specific control adapted to different delay function. 
We prove the Lipschitz continuity of the control operator across distinct delay-dependent regions, ensuring the trained DeepONet approximates the feedback law. We also establish semi-global practical stability of the closed-loop system. Numerical results demonstrate that the DeepONet-based controller achieves an approximation loss on the order of $10^{-4}$ while reducing computation time by at least an order of magnitude. We also evaluate the controller under noisy delay inputs, confirming its robust performance.

The paper is structured as follows: Section \ref{Backstepping-control} reviews delay-compensated controllers using the backstepping method. Section \ref{DeepONet based Controller} proves the control operator's Lipschitz continuity, ensuring Neural Operator (NO) approximation. Section \ref{Stability-feedback} establishes semiglobal practical stability. Section \ref{Numerical Simulations} presents numerical experiments, and conclusions are in Section \ref{Conclusions}.
 \textbf{Notation:} 
 Define sets $\mathcal T_1 =\left\{(s,q)\in \mathbb{R}^2:0\leq s \leq q \leq 1\right\}$ and $\mathcal T_2 =\left\{(s,q)\in \mathbb{R}^2:0\leq s,q\leq 1\right\}$. For $z(s) \in L^\infty[0,1]$, $ f_1(s,q) \in  L^\infty (\mathcal{T}_1) $ and $ f_2(s,q) \in  L^\infty (\mathcal{T}_2) $, define the norms
 \begin{align*}
        \lVert z \rVert_\infty&:=\sup_{s\in [0,1]}| z(s)|,  ~~\lVert f_1 \rVert_\infty:=\sup_{(s,q)\in \mathcal{T}_1}| f_1(s,q)|,  \\ \|f_2\rVert_\infty&: =\sup_{(s,q)\in \mathcal{T}_2}| f_2(s,q)|.
\end{align*}

\section{Backstepping Control for Spatially-Varying State Delay Systems} \label{Backstepping-control}
  We consider the system \eqref{eq:main-x0}-\eqref{eq:x-before} with spatially-varying state delay $\tau(s)>0$ and $U(t)$ is the control input,  which will be determined subsequently.

  \begin{assump}\label{ass-1}
 Assume the delay function $\tau(s)\in \mathcal{D}$, where 
 \begin{align}
        \mathcal{D} =& \left\{ \tau \in C^2[0,1]: \tau(s) > 0~ ~\text{for} ~s\in [0,1] \right. \nonumber \\&\left.~~~~~\text{and ~if}~ \tau(s)<s, ~ \tau'(s) < 1  \right\}.\label{set:D}
 \end{align}
\end{assump}
 \begin{assump}\label{ass-2}
   Assume the coefficient functions $c \in { C}^1 [0,1]$ with  $c(1)=0$, and $f\in {C}^1( {\cal T}_1)$.
\end{assump} 

From the above two assumptions, we can specify the following bounds:
  \begin{itemize}
             \item $\bar{\tau} = \|\tau(s)\|_\infty$, ~~~~$\bar{\tau}' = \|\tau'(s)\|_\infty$,
      \item $\bar c =\|c(s)\|_\infty$,~~~ 
         $\bar f  =\|f(s,q)\|_\infty$.
  \end{itemize}

  \begin{rmk} \label{rmk_1}
(1) Given $\tau \in C^2[0,1]$, there exist $L_{\tau},~L_{\tau'}>0$ such that $\tau(s)$ and $\tau'(s)$ are Lipschitz continuous, 
  \begin{align}
        |\tau(s_1) - \tau(s_2)| &\le L_{\tau} |s_1 - s_2|,\label{ieq:Lips_tau}
         \\|\tau'(s_1) - \tau'(s_2)| &\le L_{\tau'} |s_1 - s_2|.\label{ieq:Lips_tau_derivative}
    \end{align}
where $L_{\tau},L_{\tau'}>0$ are Lipschitz constants.

(2) 
Defining an auxiliary function for the case $\tau(s)<s$, 
\begin{equation}\label{eq:g}
    g(s):=s-\tau (s), ~~0\le s\le 1
\end{equation}  
and let $\bar{g} :=  \|g(s)\|_\infty$. Since $\tau'(s)<1$, we have $g'(s)=1-\tau'(s)>0$, hence $g$ is monotonically increasing on $[0,1]$ and admit the inverse  $g^{-1}$ defined on $[g(0), g(1)]$. Denote $\bar{g}' = \| g'(s)\|_\infty$ and $\underline g' = \inf_{s \in [0,1]} g'(s)>0$. Then $|(g^{-1})'(\sigma)| = \frac{1}{g'(g^{-1}(\sigma))} \le \frac{1}{\underline g'}$, which gives $g^{-1}(\sigma)$ is Lipschitz with constant $L_g=1/\underline g'$:
\begin{align}\label{Lip_Lg}
    |g^{-1}(\sigma_1)-g^{-1}(\sigma_2)|\leq L_g |\sigma_1-\sigma_2|.
\end{align}
\end{rmk}

\begin{rmk} 
   Based on Assumption \ref{ass-2}, we know $c$ and $f$ are Lipschitz continuous with  
   \begin{align}
   \left|c(s_1)-c(s_2)\right|&\le L_c |s_1-s_2|,\label{Lip_Lc}\\
       \left|f(s_1,\cdot)-f(s_2, \cdot)\right|&\le L_f |s_1-s_2|,\label{Lip_Lf}
   \end{align}
   where $L_c, L_f>0$ are Lipschitz constants. 
\end{rmk}

We introduce a $2$-D transport PDE with spatially-varying transport speed   to  hide the delay, which gives
  \begin{align}
  	\partial_t x (s,t)&=-\partial_s x (s,t)+c(s)u(s,0,t) \nonumber\\ &\quad+\int^1_{s}\!  f(s,q)x(q,t)dq ,\label{eq:main-x1} \\
  	x(0,t)&=  U(t),\label{eq:bnd-x1}\\
  	\tau(s)\partial_t u (s,r,t)& = \partial_r u (s,r,t),~(s,r)\in \mathcal{T}_2,
   \label{eq:main-u1}  \\
  	u(s,1,t)&=x(1,t), \label{eq:Combine-bnd-u1}\\
      x(s,0)&=  x_0(s),\label{eq:initial-x1}\\
  	u(s,r,0)&=u_0(s,r).  \label{eq:initial-u}
  \end{align} 
  We sketch the backstepping design with state feedback for system \eqref{eq:main-x1}-\eqref{eq:Combine-bnd-u1}. 
The backstepping transformation splits into two cases, for $\bar g \leq s \leq 1$,
 \begin{align}
	&z(s,t)=x(s,t) -\int^1_s  K(s,q)
	x(q,t)d q\nonumber \\
	&\quad-\int_{s}^{1}\! \int_0^{\frac{q-s}{\tau(q)}} c(q)\tau(q) K(s+\tau(q)p,q)u(q,p,t)d pd q\nonumber \\
	& \quad +\int_s^1 c(q)u\left(q,\frac{q-s}{\tau(q)},t\right)d q,  \label{eq:trans-simplify1}
\end{align}
and for $0 \leq s < \bar g$, 
\begin{align}
	&z(s,t)=x(s,t) -\int^1_s \! K(s,q)
	x(q,t)d q\nonumber \\
	&-\int_{g^{-1}(s)}^{1}\! \int_0^1  \! c(q)\tau(q) K(s+\tau(q)p,q)u(q,p,t)d pd q\nonumber \\
	&-\int_{s}^{g^{-1}(s)}\! \int_0^{ \frac{q-s}{\tau(q)}}  \! c(q)\tau(q) K(s+\tau(q)p,q)u(q,p,t)d pd q\nonumber \\
	& +\int_{s}^{g^{-1}(s)}  c(q) u\left(q,\frac{q-s}{\tau(q)},t\right) dq. \label{eq:trans-simplify2}
\end{align}
Applying the transformation, we get the following stable target system
   \begin{align}
 	\partial_t z (s,t)&=-     \partial_s  z(s,t),\label{eq:main-tar-z} \\       	z(0,t)&=0,\label{eq:bnd-tar-z}\\
 	\tau(s) \partial_t u (s,r,t) &= \partial_r u (s,r,t) , \label{eq:tar-u} \\
 	u(s,1,t)&=z(1,t), \label{eq:bnd-tar-u}
 \end{align}
 which gives 
  \begin{equation}\label{solution-tar-u}
  	u(s,r,t)=\left\{ 
  	\begin{aligned}
  		&u_0(s,r+t/\tau(s)),&t<\tau(s)(1-r),\\
  		&z(1,t-\tau(s)(1-r)), &t\geq\tau(s)(1-r).
  	\end{aligned}\right.
  \end{equation}
To map \eqref{eq:main-x1}-\eqref{eq:Combine-bnd-u1} into \eqref{eq:main-tar-z}-\eqref{eq:bnd-tar-u} the kernel function should satisfy: 
 \begin{align}\label{eq:k_integral}
	\partial_s K+\partial_qK&=f(s,q)
	  -\int_{s}^{q} K(s,r)f(r,q)d r, 	
\end{align}
with boundary conditions 
\begin{align}
    K(s,1)&=0,\label{eq:K-bnd-case1} ~~\textrm{for}~~  \bar g \leq s,
    \\
		K(s,1)&= \int_{g^{-1}(s)}^1 c(p)K(s+\tau(p),p) dp-\frac{c(g^{-1}(s))}{g'(g^{-1}(s))},\nonumber
  \\ & ~~~~~\text{for}~~    s <  \bar g,\label{eq:K-bnd-case2}
\end{align} 
Applying the characteristic method, we obtain the integral form,
\begin{align}\label{K-integral_equation}
     K(s,q)=&~K(s-q+1,1)-\int_{q}^{1}f(\theta+s-q,\theta)d\theta\\
	&~+ \int_s^{s+1-q}\int_\theta^{\theta-s+q} K(\theta,r)f(r,\theta-s+q)drd\theta.\nonumber
\end{align}
Substituting the boundary conditions \eqref{eq:K-bnd-case1}
 and \eqref{eq:K-bnd-case2} into  \eqref{K-integral_equation}, we get 
\begin{align}\label{eq:kernel_K}
     K(s,q)=  
\begin{cases} 
 K_1(s,q),&\text{if } ~q-s \leq \tau(1), \\
K_2(s,q) , &\text{if }~ q-s > \tau(1),  
\end{cases}
\end{align} 
where $0\le s\le q\le 1$ and
\begin{align}\label{eq:K_case1} 
    \sigma(s,q)=&~s+1-q,\\
    K_1 = &~ \Psi _1(K_1) -\Xi_1,\\  
    K_2= & ~ \Psi_1(K) -\Xi_1- \Xi_2+\Psi_{21}(K_1) +\Psi_{22}(K_2), \label{eq:K_case2}   
 \end{align}
with
\begin{align}
    \Psi_1(K)(s,q)&=\int_{s}^{\sigma}\int_{\theta}^{\theta -s+q} K(\theta, r) f(r, \theta-s+q)   dr   d\theta ,\label{eq:Psi1}\\
     \Psi_{21}(K_1)(s,q) &= \int_{g^{-1}(\sigma)}^{\psi(s,q,\bar g)} c(p) K_1(\sigma+\tau(p), p) d p,\label{eq:Psi21}\\
     \Psi_{22}(K_2)(s,q)& = \int_{\psi(s,q,\bar g)}^1 c(p) K_2(\sigma+\tau(p), p) d p,\label{eq:Psi22}\\
     \psi(s,q,\bar g)&=g^{-1}\left( \min\{\bar g, ~\sigma+\tau(1)\}\right),\label{eq:psi}\\
     \Xi_1(s,q)&  =\int_{q}^{1}f(\theta+s-q, \theta)   d\theta,\label{eq:Xi1}
  \\\Xi_2(\sigma) & =\frac{c(g^{-1}(\sigma))}{g'(g^{-1}(\sigma))\label{eq:Xi2}}. 
\end{align}
The existence and boundedness of the kernel function have been proved in \cite{zhang2021compensation} and \cite{ZHANG2024105964}. Specifically, we represent the upper bound of the kernel function by $\bar{K}:=\|K\|_\infty$. 
Based on the  boundary conditions \eqref{eq:bnd-x1} and \eqref{eq:bnd-tar-z}, along with the transformation \eqref{eq:trans-simplify1} and \eqref{eq:trans-simplify2}, the controller is derived
  \begin{align}
\nonumber 	U(t)= &\int^1_0 \! K(0,q)    x(q,t)dq-\int_0^1 c(q)u\left(q,\frac{q}{\tau(q)},t\right)dq \\  \nonumber 
	& +\int_0^{1}\! \int_0^{q}c(q)   K( p,q)u\left(q,\frac{p}{\tau(q)},t\right)dpdq, \\
  &\text{for}~   \tau \in \mathcal{D}_1,  \label{eq:U-case1}  
 \\U(t)= &\int^1_0 \! K(0,q)    x(q,t)dq  -\int_{0}^{g^{-1}(0)}  c(q) u\left(q,\frac{q}{\tau(q)},t\right) dq    \nonumber \\
	 & +\int_{0}^{1}\! \int_0^{\min\{\tau(q),q\}} \! c(q) K(p,q)u\left(q,\frac{p}{\tau(q)},t\right) dp dq, \nonumber 
 \\&\text{for}~  \tau \in \mathcal{D}_2,  \label{eq:U-case2}
\end{align}   
where 
\begin{align}
    \mathcal{D}_1=&\{\tau \in C^2[0,1] ~|~ 
     \tau(1) \ge 1 \},   \label{domain_tau1}\\
    \mathcal{D}_2=&\{\tau \in C^2[0,1]~|~  \tau(1) <1 \}.  \label{domain_tau2}
\end{align}
If $\tau(q) \in \mathcal{D}_1$, the kernel $K(s, q)$ is determined by \eqref{eq:K_case1}. Conversely, if $\tau(q) \in \mathcal{D}_2$, $K(s, q)$ is governed by both \eqref{eq:K_case1} and \eqref{eq:K_case2}, indicating controller in this case involves two types of kernel gains. 
Note that the transformation  at $s=\bar g$ and the controllers at $\bar g=0$ are continuous due to $g^{-1}(\bar g)=1$. 
The inverse transformation is \cite{zhang2021compensation}: 
\begin{align}
	\nonumber x(s,t)=&
	 z(s,t)+\int_{s}^{1}{{F_1}(s,q) z(q,t)}dq\\
	 &+\int_{0}^{1}{\int_{0}^{1}} F_2(s,q,r)u(q,r,t)drdq.\label{eq:inverse_trans}
\end{align}
Recall that the well-posedness of kernel functions $F_1$ and $F_2$ is stated in Theorem 2 of paper \cite{zhang2021compensation}. 

\section{Approximation of the Neural Operator Controller with DeepONet}\label{DeepONet based Controller}
The DeepONet approximation theorem \cite{lu2021learning} provide the theoretical basis for using DeepONet-based controllers. This theorem stipulates the operator approximated by the DeepONet must be continuous, with Lipschitz continuity specially enabling the estimation of approximation errors based on the network's parameters.
\begin{definition}\label{def:kernel-operator}
   The kernel operator  $\mathcal{K}: \mathcal{D} \mapsto C^0(\mathcal{T}_1)$ with
\begin{align}
  K (s,q)&=:\mathcal{K}(\tau)(s,q),\label{eq:define_opeK}
\end{align} 
is defined by \eqref{eq:kernel_K}.
Specifically, for $q-s>\tau(1)$, define the operator 
\begin{equation}
    K_2(s,q)=:\mathcal{K}_2(\tau)(s,q), \label{eq:define_opeK2}
\end{equation}
with $K_2(s,q)$ defined in \eqref{eq:K_case2}.
\end{definition}

Note that this paper focuses exclusively on operators dependent on $\tau$, as the operator that maps $f$ and $c$ has been addressed in \cite{qi2024neural} and including $f$ and $c$ here would not increase the technical difficulty. By isolating $\tau$, we can specifically analyze its role in the kernel and control operators, as $\tau$ not only appears in the kernel function and the integration limits, but also delineates the spatial regions where the kernel assumes different forms. 
 
Appendix \ref{C1stability1} demonstrates that the closed-loop system is stable in the $C^1$ norm when using the controller defined by equations \eqref{eq:U-case1} and \eqref{eq:U-case2}. It implies that $x\in C^1[0,1]$ and $u\in C^1([0,1]^2)$ and we define 
\begin{definition}\label{def:ctr_ope}
 The controller operator $\mathcal{U}: \mathcal{D}\times C^1[0,1]  
 \times C^1([0,1]^2)\mapsto \mathbb{R}$ with  	 
 \begin{align}\label{ope-U}
 U=\mathcal{U}( \tau, x,u),
 \end{align} 
is defined by the expressions \eqref{eq:U-case1} and \eqref{eq:U-case2}.
\end{definition}

\begin{lemma}\label{Lemma:inverse_g_0}
       Under Assumption \ref{ass-1}, the following inequality holds 
        \begin{align}
            &\|g_1^{-1}(\sigma)-g_2^{-1}(\sigma)\|_\infty\leq \frac{\|\tau_1-\tau_2\|_\infty}{\underline g'},~\tau_1, \tau_2 \in  \mathcal{D}_2,\label{ieq-g-inverse_same_section}\\ 
          &  \|1-g_2^{-1}(0) \|_\infty \leq \frac{\|\tau_1-\tau_2\|_\infty}{\underline g'}, ~\tau_1 \in \mathcal{D}_1, ~\tau_2 \in \mathcal{D}_2\label{ieq-g-inverse_diff_section},
        \end{align}
        where $g_i^{-1}(\sigma)$, $i=1,2$ represents the inverse function of $g(q)=q-\tau_i(q)$ dependent on $\tau_i$.
\end{lemma}
    
    The proof can be found Appendix \ref{proof:lemma1_2}.
    
Note that $\mathcal{K}(\tau)(s,q)$ maps $\tau$ to $K_1(p,q)$ when $q-s \le \tau(1) $, and to $\mathcal{K}_2(\tau)(s,q)=K_2(s,q)$ when $q-s >\tau(1)$. Since  $K_1(p,q)$ itself is independent of $\tau$, it is suffices to prove the Lipschitz continuity of $\mathcal{K}_2(\tau)$ with respect to $\tau$ as follows.
\begin{lemma}\label{lemma:lip_k_tau_0}
For $\tau_1,\tau_2 \in \mathcal{D}_2$, the operator $ \mathcal{K}_2(\tau)$ defined in \eqref{eq:define_opeK2} exhibits Lipschitz continuity, satisfying 
    \begin{align}\label{ieq:Ktau1-Ktau2}
        \| \mathcal{K}_2(\tau_1)- \mathcal{K}_2(\tau_2)\|_\infty\le L_{K}\|\tau_1-\tau_2\|_\infty,
    \end{align}
where 
\begin{align}
    L_{K}=&L_{\Phi_0} \mathrm{e}^{\max\{\bar f,\bar c\}},\\
L_{\Phi_0}=&L_F+2\bar K\bar f + \frac{3\bar c\bar K}{\underline g'}+\bar c(L_1+L_2)
    +4L_g\bar c\bar K,\\
        L_F=&\frac{\bar cL_B+\bar cL_{\tau'}+L_c(1-\bar\tau')}{(1-\bar \tau')^2},\\
 L_1=&3\bar f\bar K+L_f(1+\bar K) ,\\
L_2=& \underline g'\mathrm{e}^{\bar c/\underline g'}\left(3\bar c\bar K L_g+L_f(1+\bar K)(1+\bar c)+\bar  3\bar f\bar K(1+\bar c)\right.\nonumber 
    \\&\left.+\frac{L_g}{(1-\bar \tau')^2}(\bar c L_{\tau'}+L_c \bar g')\right).
 \end{align}    
Here, $L_{\tau}$, $L_{\tau'}$, $L_g$, $L_c$ and $L_f$ are Lipschitz constants of $\tau(q)$, $\tau'(q)$, $g^{-1}$, $c(s)$ and $f(s,\cdot)$, whose Lipschitz inequalities are defined in \eqref{ieq:Lips_tau}, \eqref{ieq:Lips_tau_derivative}, \eqref{Lip_Lg}, \eqref{Lip_Lc} and \eqref{Lip_Lf}, respectively. 
\end{lemma}  

 The proof is detailed in Appendix \ref{proof:lemma1_2}.
 
\begin{lemma}\label{lemma-ULip}
	(Lipschitzness of the control operator). The control operator $\mathcal{U}: \mathcal{D}\times C^1[0,1]  \times C^1([0,1]^2)\mapsto \mathbb{R}$  in Definition \ref{def:ctr_ope} 
 is Lipschitz continuous and satisfies
	\begin{align}\label{eq:U_Lip}
		&|\mathcal{U}(\tau_1, x_1, u_1)-\mathcal{U}(\tau_2,x_2, u_2)|\nonumber \\ \leq &L_{\mathcal{U}}\max\{\|\tau_1-\tau_2\|_\infty, \|x_1-x_2\|_\infty, \|u_1-u_2\|_\infty \},  
	\end{align}
with the Lipschitz constant 
 \begin{align}
     L_{\mathcal{U}}=\max&\left\{\bar K,\bar c(1+\bar K),6\bar c\bar u \bar K+2\bar K \bar x +\bar cL_u(1+\bar K) \right.\nonumber \\& \left.~+L_{K}(\bar x+\bar c\bar u) +   \frac{\bar c\bar u } {\underline g'} \right\}.
 \end{align}
\end{lemma}
\begin{proof}
To facilitate our analysis of the Lipschitz continuity of $u$ with respect to $\tau $, we simplify the notations of both $u\left(q,\frac{p}{\tau(q)}\right)$ and $u\left(q,\frac{q}{\tau(q)}\right)$ to $u(\tau)$. Given $u\in C^1 ([0,1]^2 )$,    it is evidence that 
\begin{align}\label{eq:u_Lips_appix}
   |u(s,r_1)-u(s,r_2)|\le L_{\breve u} |r_1 - r_2|,~L_{\breve u}>0.
\end{align}
Given  $\tau \neq 0$,  we have 
\begin{align}\label{ieq:u_Lip_appix}
    \|u(\tau_1)-u(\tau_2)\|_{\infty}\le& L_{u} \|\tau_1 - \tau_2\|_{\infty},~L_{u}>0.
\end{align}
We discuss the Liptchitz continuity in the following three cases: $\tau_1, ~\tau_2 \in \mathcal{D}_1$, and $\tau_1\in \mathcal{D}_1,~\tau_2 \in \mathcal{D}_2$, as well as $\tau_1, ~\tau_2 \in \mathcal{D}_2$.

Before proceeding, we denote $U_1=\mathcal{U}(\tau_1, x_1, u_1)$ and $U_2=\mathcal{U}(\tau_2,x_2, u_2)$.

\textbf{Case 1}: 
As $\tau_1, ~\tau_2 \in \mathcal{D}_1$, the kernel function $K(s,q)$ defined in \eqref{eq:K_case1} is independent of $\tau$, which gives 
\begin{align}\label{ieq:caseA_appix}
  &~|U_1-U_2|\nonumber\\
  = &~\left|\int^1_0 \! K(0,q)(x_1-x_2)(q)dq\right.\nonumber \\
 &~ \left.-\int_0^1 c(q)\left[u_1(\tau_1)-u_2(\tau_2)\right]dq \right.\nonumber \\
& ~  \left.+\int_0^{1}\! \int_0^{q}c(q)   K( p,q)\left[u_1(\tau_1)-u_2(\tau_2)\right]dp dq\right|\nonumber \\
  \leq&~  \bar{K}\|x_1-x_2\|_{\infty}+\bar c \|u_1-u_2\|_{\infty} +\bar c L_u \| \tau_1 - \tau_2 \|_{\infty}\nonumber \\&~+\bar c\bar{K}  \|u_1-u_2\|_{\infty}+\bar c \bar{K} L_u \|\tau_1-\tau_2\|_{\infty}\nonumber \\&~\leq
 \bar{K}\|x_1-x_2\|_{\infty}+\bar c(1+\bar{K}) \|u_1-u_2\| _{\infty}\nonumber \\& ~+\bar cL_u(1+\bar{K})\|\tau_1-\tau_2\|_{\infty},
\end{align}  
where we use the Lipschitz condition \eqref{ieq:u_Lip_appix}.

\textbf{Case 2}: $\tau_1\in \mathcal{D}_1,~\tau_2 \in \mathcal{D}_2$. Denote $g^{-1}(\tau_2)$ simply as $g^{-1}$ and let 
\begin{align}
\label{eq:U1-U2_B_appix}
      |U_1-U_2| =|\Delta_1+\Delta_2+\Delta_3|,
\end{align}
where
\begin{align}
    \Delta_1 =&\int^1_0 \! (\mathcal{K}(\tau_1)(0,q)x_1(q)-\mathcal{K}(\tau_2)(0,q)x_2(q))dq,\\
     \Delta_2=&-\int_0^1 c(q) u_1(\tau_1)dq+\int_0^{g^{-1}(0)}c(q)u_2(\tau_2)dq,\\
     \Delta_3=&\int_0^1\int_0^q c(q)\mathcal{K}(\tau_1)(p,q)u_1(\tau_1)dpdq
        \\& -\int_{0}^1\int_0^{\min\{\tau_2(q),q\}} c(q)\mathcal{K}(\tau_2)(p,q)u_2(\tau_2)dp dq. \nonumber
 \end{align}
Note that $\mathcal{K}(\tau_1)(s,q)=K_1(s,q)$ due to $\tau_1>1$.
We start with the first term 
\begin{align}
    |\Delta_1| \le & \left|\int_0^1 K_1(0,q)(x_1-x_2)dq \right|\nonumber \\
    &+ \left|\int_{\tau_2(1)}^1 (K_1(0,q)-K_2(0,q))x_2(q)dq\right|  
    \nonumber \\
   \le  &2\bar K\bar x\int_{\tau_2(1)}^{\tau_1(1)}dq +   \bar K\|x_1-x_2\| _{\infty}\nonumber 
    \\\le &  2\bar K\bar x\|\tau_1-\tau_2\|_{\infty}+\bar K\|x_1-x_2\|_{\infty}.
    \label{ieq:Delta1_appix}
\end{align}
  where we use  $\tau_1(1)\ge 1$ to derive the second line from the first line of \eqref{ieq:Delta1_appix}. 
\begin{align*}
    |\Delta_2|\le &\left|-\int_0^1 c(q) u_1(\tau_1)dq+\int_0^{g^{-1}(0)} c(q) u_1(\tau_1)dq\right|\\&+\left|-\int_0^{g^{-1}(0)} c(q) u_1(\tau_1)dq+\int_0^{g^{-1}(0)} c(q) u_2(\tau_1)dq
    \right|\\&
   +\left|-\int_0^{g^{-1}(0)} c(q) u_2(\tau_1)dq+\int_0^{g^{-1}(0)} c(q) u_2(\tau_2)dq\right|\\
   \le &\left(\frac{\bar c\bar u}{\underline g'}+\bar cL_u\right)\|\tau_1-\tau_2\|_{\infty}+\bar c \|u_1-u_2\|_{\infty},
\end{align*}
where we use the second inequality of Lemma \ref{Lemma:inverse_g_0}.
\begin{align}\label{Case2_Delta_3_appix}
     \Delta_3 = &  \int_0^1\int_0^q c(q)K_1(p,q)(u_1(\tau_1)-u_2(\tau_2))dpdq
       \\& +\int_{g^{-1}(0)}^1\int_{\tau_2(q)}^{q} c(q)K_1(p,q)u_2(\tau_2)dpdq
      \nonumber \\&  +\int_{\tau_2(1)}^1 \int_0^{\phi(q)}c(q)(K_1 -K_2)(p,q) u_2(\tau_2)dpdq, \nonumber
\end{align}
    where  $\phi(q)=\min\{q-\tau_2(1),\tau_2(q)\}$.
Therefore,
    \begin{align}
    |\Delta_3|\le & \bar c\bar K L_u \|\tau_1-\tau_2\|_{\infty}+\bar c\bar K \|u_1-u_2\|_{\infty}\\&+\frac{\bar c\bar K \bar u }{\underline g'}\|\tau_1-\tau_2\|_{\infty}+2\bar c\bar K \bar u\|\tau_1-\tau_2\|_{\infty}\\
    \le &
    \bar c\left(\bar KL_u+2\bar K\bar u +\frac{\bar u\bar K}{\underline g'}  \right) \|\tau_1-\tau_2\|_{\infty}\nonumber \\& + \bar c\bar K\|u_1-u_2\|_{\infty}.
    \end{align}     
Finally, we reach 
\begin{align}\label{ieq:caseB_appix}
    &|U_1-U_2|\le \bar{K}\|x_1-x_2\| _{\infty}+\bar c(1+\bar{K}) \|u_1-u_2\|_{\infty}\\
 \nonumber& ~+\left[\bar c(1+\bar K)\left(\frac{\bar u}{\underline g'}+L_u \right)+   2\bar K(\bar x+\bar c\bar u)\right] \|\tau_1-\tau_2\|_{\infty}.
\end{align} 
\textbf{Case 3}:  We use $g_1^{-1}$ and $g_2^{-1}$ to simplify the notation of the inverse function of $g$ for $\tau_1 \in \mathcal{D}_1$ and $\tau_2\in \mathcal{D}_2$, respectively. 
Let
\begin{align}
    U_1-U_2=\Delta_1+\Delta_2+\Delta_3,
\end{align}
where 
\begin{align}
      \Delta_1 =&\int^1_0 \! \mathcal{K}(\tau_1)(0,q)x_1(q)-\mathcal{K}(\tau_2)(0,q)x_2(q)dq,\\
     \Delta_2=&-\int_0^{g_1^{-1}(0)} c(q) u_1(\tau_1)dq+\int_0^{g_2^{-1}(0)}c(q)u_2(\tau_2)dq,\\
      \Delta_3=&\int_{0}^1\int_0^{\min\{\tau_1(q),q\}} c\mathcal{K}(\tau_1) u_1(\tau_1)dpdq
     \nonumber \\&
     -\int_{0}^1\int_0^{\min\{\tau_2(q),q\}} c\mathcal{K}(\tau_2) u_2(\tau_2)dpdq.
\end{align}
Recalling that $\mathcal{K}(\tau)$ maps $\tau$ to $K_2(p,q)$ for $p < q-\tau(1)$, and to $K_1(p,q)$ for $p \ge q-\tau(1) $, we rewrite the integration region of $\Delta_1$ as
\begin{align}
     \Delta_1 =&\int_0^{\tau_1(1)} K_1(0,q)x_1(q) dq-\int_{\tau_1(1)}^1 \mathcal{K}_2(\tau_1)(0,q)x_1(q) dq \nonumber
    \\&-\int_0^{\tau_2(1)} K_1(0,q)x_2(q) dq\nonumber \\ &+\int_{\tau_2(1)}^1 \mathcal{K}_2(\tau_2)(0,q)x_2(q) dq,\nonumber
    \end{align}
which gives
\begin{align}
  |\Delta_1|\le &2\bar K\bar x \|\tau_1-\tau_2\|_{\infty}+\bar     x\|\mathcal{K}_2(\tau_1)-\mathcal{K}_2(\tau_2)\|  _{\infty} 
    \nonumber \\ &+\bar K\|x_1-x_2\| _{\infty}
    \nonumber\\
    \le &\left( 2\bar K\bar x  +\bar xL_{K}\right)\|\tau_1-\tau_2\| _{\infty}+\bar K\|x_1-x_2\| _{\infty}.\label{ieq:caseC_1_appix}      
\end{align}
Consider $u\left(q,\frac{q}{\tau(q)}\right)$ defined on $[0,1]^2$ with the condition  $q\le \tau(q)$. To ensure that  $u$ integrates within this domain, we derive 
 \begin{align}
     |\Delta_2|\le &\left|\int_0^{\min_{i=1,2}\{g^{-1}_i(0)\}}c(u_1(\tau_1)-u_2(\tau_2))dq\right|\nonumber \\
     &+\bar c\bar u\left|\max_{i=1,2}\{g^{-1}_i(0)\}-\min_{i=1,2}\{g^{-1}_i(0)\}\right|\nonumber \\ \le & \bar c\left(L_u+\frac{\bar u}{\underline g'}\right)\|\tau_1-\tau_2\|_{\infty} +\bar c\|u_1-u_2\|_{\infty} .\label{ieq:caseC_2_appix}   
 \end{align}
Similarly, for $u(q,\frac{p}{\tau(q)})$ defined on $[0,1]^2$ with  $p\le \tau(q)$, we have
\begin{align*}
    |\Delta_3|\le & \left|\int_{0}^1\int_{\psi_1(q)}^{\min\{\tau_1(q),q\}} c\mathcal{K}(\tau_1) u_1(\tau_1)dpdq \right.
    \nonumber 
    \\&\left.+\int_{0}^1\int_0^{\psi_1(q)} c(\mathcal{K}(\tau_1) u_1(\tau_1)-\mathcal{K}(\tau_2) u_2(\tau_2))dpdq\right.
    \nonumber 
    \\& \left.
    \int_{0}^1\int_{\psi_1(q)}^{\min\{\tau_2(q),q\}} c\mathcal{K}(\tau_2) u_2(\tau_2)dpdq\right|\nonumber 
    \\\le & 
    \left|\int_\eta^1
     \int_0^{\psi_2(q)} c(\mathcal{K}_2(\tau_1)-\mathcal{K}_2(\tau_2))u_1(\tau_1)dpdq \right.\nonumber 
     \\&  \left.+\int_{\tau_1(1)}^\eta\int_0^{\psi_1(q)}c\mathcal{K}_2(\tau_1)u_1(\tau_1)dpdq\right.
     \nonumber 
     \\&      \left.+\int_\eta^1\int_{\psi_2(q)}^{\psi_1(q)}c\mathcal{K}_2(\tau_1)u_1(\tau_1)dpdq\right.
     \nonumber 
     \\&  \left.+\int_{\tau_2(1)}^\eta\int_0^{\psi_1(q)}c\mathcal{K}_2(\tau_2)u_1(\tau_1)dpdq\right.
     \nonumber 
     \\&      \left.+\int_\eta^1\int_{\psi_2(q)}^{\psi_1(q)}c\mathcal{K}_2(\tau_2)u_1(\tau_1)dpdq\right.
     \nonumber 
     \\&     \left. +\int_0^1\int_{0}^{\psi_1(q)}c\mathcal{K}(\tau_2)(u_1(\tau_1)-u_2(\tau_2))dpdq \right|
     \nonumber 
     \\&+2\bar c \bar K \bar u  \|\tau_1-\tau_2\|_{\infty}
     \nonumber \\     
        \le & \bar c \bar u\|\mathcal{K}_2(\tau_1)-\mathcal{K}_2(\tau_2)\|_{\infty}+\bar c\bar K \|u_1-u_2\|_{\infty}\nonumber
        \\&+(\bar c \bar K L_u+6\bar c\bar K\bar u)\|\tau_1-\tau_2\|_{\infty} 
             \nonumber \\     
        \le &  \bar c (\bar u L_K+\bar K L_u +6 \bar K\bar u)\|\tau_1-\tau_2\|_{\infty}+\bar c\bar K \|u_1-u_2\|_{\infty} ,
\end{align*}
where $\psi_1(q)=\min\{\tau_1(q),\tau_2(q),q\}$, $\psi_2(q)=\min\{q-\tau_1(1),q-\tau_2(1), \tau_1(q),\tau_2(q)\}$, $\eta=\min\{\tau_1(1),\tau_2(1)\}$. 
Note that  $\mathcal{K}(\tau)$ maps $\tau$ to either $K_1$ or $K_2=\mathcal{K}_2(\tau)$, where $K_1$ is independent of $\tau$, causing its integral to cancel out. The integration area of $\mathcal{K}(\tau_1)-\mathcal{K}(\tau_2)$ and $\mathcal{K}_2(\tau_1)-\mathcal{K}_2(\tau_2)$ are shown in Fig. \ref{fig:integration_area}.
\begin{figure}[htbp]
    \centering
    \includegraphics[width=0.4\textwidth]{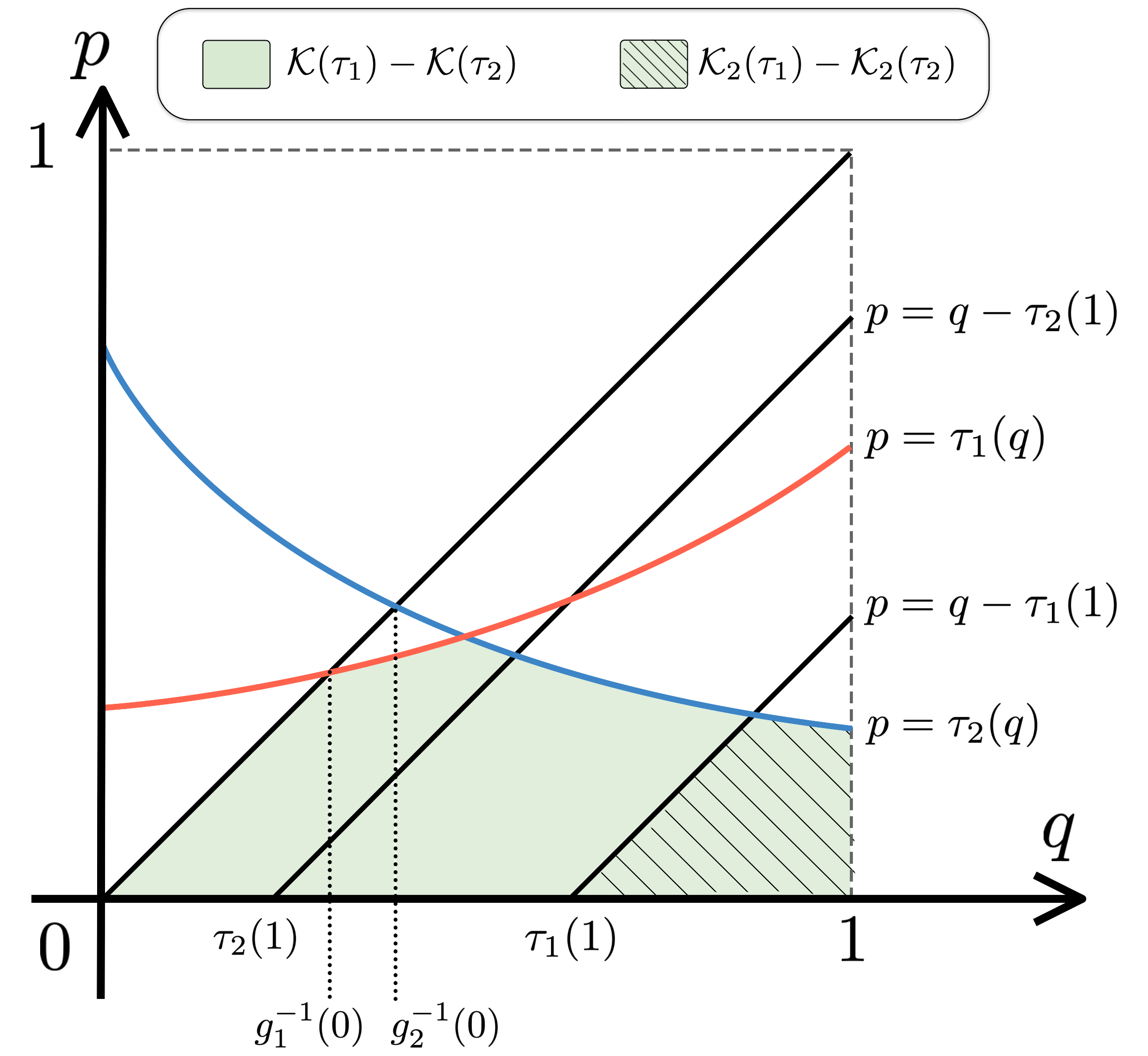}
    \caption{ Integration area of $\mathcal{K}(\tau_1)-\mathcal{K}(\tau_2)$ and $\mathcal{K}_2(\tau_1)-\mathcal{K}_2(\tau_2)$.}
    \label{fig:integration_area}
\end{figure}

Consequently,
 \begin{align}
     |U_1-U_2| 
 \leq &~
 \left[
 6\bar c\bar u \bar K+2\bar K \bar x +\bar cL_u(1+\bar K)+L_{K}(\bar x+\bar c\bar u)  \right.\nonumber  \\& ~+ \bar K\|x_1-x_2\|_{\infty}+\bar c(1+\bar K)\|u_1-u_2\|_{\infty}\nonumber\\& \left.~+   \frac{\bar c\bar u } {\underline g'}\right] \|\tau_1-\tau_2\|_{\infty}.
\end{align}
 Combining the three cases, we obtain the Lipschitz constant $L_{\mathcal{U}}$ and reach inequality \eqref{ieq:u_Lip_appix}.
\end{proof}

According to the DeepONet approximation theorem from [\cite{deng2022approximation}, Th. 2.7 and Remark 2.8],  Lemma \ref{lemma-ULip}, we obtain the following result by instantiating with $d=2$ and $\alpha=1$. 

 Let $(\tau, x, u) \in C^2[0,1] \times C^1[0,1] \times C^1([0,1]^2)$ with $\|\tau\|_\infty \leq B_{\tau}$, $\|x\|_\infty \leq B_x$, and $\|u\|_\infty \leq B_u$ be the inputs of the control operator \eqref{def:ctr_ope}, which is discretized as $(\bm{\tau}, \bm{x}, \bm{u})_m = [(\tau, x, u)_1, \ldots, (\tau, x, u)_m]^T$ on $[0,1]^2$ with grid size $m$.
  \begin{theorem}(DeepONet approximation theorem) \label{th:error_NO}
\rm{For any approximation error $\epsilon > 0$, there exists  $p^*(\epsilon), m^*(\epsilon) \in \mathbb{N}$  such that for all  $p > p^*(\epsilon)$ and $m > m^*(\epsilon)$, there exist neural networks $f^{\mathcal{N}}(\cdot; \theta^{(k)})$ and $g^{\mathcal{N}}(\cdot; \vartheta^{(k)})$,  $k = 1, \dots, p$  satisfying the approximation bound
\begin{align}
        &\left|\mathcal{U}(\tau,x,u) - \hat{\mathcal{U}}((\bm{\tau}, \bm{x}, \bm{u})_m)\right| < \epsilon, \label{eq:error_bound}\\
        &\hat{\mathcal{U}} = \sum_{k=1}^{p} g^{\mathcal{N}}((\bm{\tau}, \bm{x}, \bm{u})_m; \vartheta^{(k)}) f^{\mathcal{N}}((s,r); \theta^{(k)}), ~ (s,r) \in [0,1]^2. \label{eq:U_NO}
    \end{align}
    }
\end{theorem}
 
\begin{rmk}
        Based on the Lipchitz continuity of the control operator established in Lemma \ref{lemma-ULip}, the DeepONet  \eqref{eq:U_NO} admits the following parameter settings with respect to the approximation error $\epsilon$: 
\begin{itemize}
    \item Discretization grid size: $m = \mathcal{O}(\epsilon^{-2})$
    \item Number of basis components: $p = \mathcal{O}(\epsilon^{-1})$
    \item Trunk network size: $|\theta^{(k)}| = \mathcal{O}\left(\left[2\ln(1/\epsilon)\right]^3\right)$ for each $k=1,\dots,p$
    \item The product of the depth (number of layers) and the width (neurons  per layer) for the branch network: $L_{g^{\mathcal{N}}} \times N_{g^{\mathcal{N}}} = \mathcal{O}(\epsilon^{-2/\epsilon})$
\end{itemize}
where $\mathcal{O}(\cdot)$ denotes the asymptotic order. 
\end{rmk}
    
\section{The Semi-global Stability under DeepONet}\label{Stability-feedback}
 As the DeepONet controller \eqref{def:ctr_ope} is applied, the resulting target system becomes \eqref{eq:main-tar-z}, \eqref{eq:tar-u},  \eqref{eq:bnd-tar-u} with boundary condition $ z(0,t)=\hat{\mathcal{U}}(\tau, x, u)-\mathcal{U}(\tau, x, u)$.

Denote $\underline B_\tau \le \tau \le B_\tau$ for each $\tau \in \mathcal{D}$.

Before presenting the main result, we first define the following Lyapunov functions
        \begin{align}
        V(t)  = &A V_{1}(t)+ V_{2}(t),\label{eq:lyapunov_stab_esti}\\
             V_{1}(t)=&\int_0^1 \mathrm{e}^{-b_1 s}|z(s,t)|^2ds, 
             \\
             V_{2}(t)=&\int_0^1\int_0^1 \tau(s)\mathrm{e}^{b_2r}|u(s,r,t)|^2drds,                
        \end{align}    
        where $b_1,~b_2$ and $A$ are positive constants.
\begin{theorem}\label{stability_theorem}
(Semiglobal practical stability under NO approximation)
\rm{ 
    For any $B_x,~B_u>0$, if  $\epsilon<\epsilon^*$, where
    \begin{align}\label{eq:epsilon*}
      \epsilon^*(\bar c, \bar f, B_\tau, \underline B_\tau, \bar K,B_x,B_u):=\sqrt{\frac{B_x^2+B_u^2}{M_2}},
    \end{align}
    and for all initial conditions that satisfy $\|x_0\|_{L^2}^2+ \|u_0\|_{L^2}^2<B_0$ with
    \begin{align}\label{eq:initial_semiglobal}
    B_0:=\frac{\beta_1}{k_1}\left(\frac{B^2_x+B^2_u}{k_2\beta_2}-A\epsilon^2\right),
    \end{align}           
the closed-loop system \eqref{eq:main-x1}-\eqref{eq:Combine-bnd-u1} under the NO-based controller $\hat{ \mathcal{U}}(\tau,x,u)(t)$ is semiglobally practically exponentially stable, satisfying the following estimate  for $\forall t>0$:
\begin{align}
    \|x(t)\|_{L^2}^2+\|u(t)\|_{L^2}^2 \leq&
   M_1\mathrm e^{-at}(\|x_0\|_{L^2}^2+\|u_0\|_{L^2}^2)+ M_2\epsilon^2, \label{ieq:stability_estimate}
\end{align}
where 
\begin{align}
A=&\mathrm{e}^{b_1 b_2}, \text{~for~any~} b_1,b_2>0, ~
a=\min\{b_1,\frac{b_2}{B_\tau}\}, \\
\beta_1=&\frac{1}{A}\min \left\{1,\frac{\mathrm e^{-b_2}}{B_\tau}\right\},~
    \beta_2= \max \left\{\mathrm e^{b_1},\frac{1}{\underline B_\tau}\right\}\label{def_beta1_2}\\
k_1=&\max\{4(1+\bar K), ~4\bar c^2 (1+B^2_\tau \bar K^2)+1\},\\
k_2 = &\max\left\{4\left(1+\bar L_1^2\right),1+4\bar c  ^2(1+B_{\tau}^2\mathrm{e}^{2\bar{f}B_{\tau}})\right\},
\\
M_1=&\frac{k_1}{\beta_1} k_2 \beta_2, ~~M_2=k_2 \beta_2 A.
\end{align}        
}
\end {theorem} 
\begin{proof}
        First, we consider the stability of the target system.        Let $\tilde{\mathcal{U}}= \mathcal{U}-\hat{\mathcal{U}}$ and from Theorem \ref{th:error_NO}, we know there exists an error $\epsilon$, such that $|z(0)|=|\tilde{\mathcal{U}}|\le \epsilon$. 
        Given constants $b_1,~b_2,~A>0$, we define the following Lyapunov functions
        \begin{align}
        V(t)  = &A V_{1}(t)+ V_{2}(t),\label{eq:lyapunov_stab_esti}\\
             V_{1}(t)=&\int_0^1 \mathrm{e}^{-b_1 s}|z(s,t)|^2ds, 
             \\
             V_{2}(t)=&\int_0^1\int_0^1 \tau(s)\mathrm{e}^{b_2r}|u(s,r,t)|^2drds,                
        \end{align}        
       Take the time derivative, 
        \begin{align*}
            \dot{V}(t)\le& -(A   \mathrm{e}^{-b_1 }-\mathrm{e}^{b_2})z^2(1,t)+A\tilde{\mathcal{U}}^{2}-Ab_1V_{1}-\frac{b_2}{B_{\tau}}V_2,  
        \end{align*} 
Let $A\mathrm{e}^{-b_1 }=\mathrm{e}^{b_2 }$, that is $A=\mathrm{e}^{b_1 b_2}$, which gives
\begin{align}
     \dot{V}(t)\le&  A\tilde{\mathcal{U}}^{2}-Ab_1V_{1}-\frac{b_2}{B_{\tau}}V_2 \le -a V+A \tilde{\mathcal{U}}^{2},
\end{align} 
where $a=\min\{b_1,\frac{b_2}{B_\tau}\}$. Using Gronwall Lemma, we have  \begin{align}\label{ieq:V_le_V0}
 V(t)\le& V(0)\mathrm{e}^{-at}+A\int_0^t \mathrm{e}^{-a(t-\breve t)}\tilde{\mathcal{U}}^2(\breve t) d\breve t
\nonumber 
\\ \le & V(0)\mathrm{e}^{-at}+\frac{A}{a} \epsilon^2  .
\end{align}

Second,  we establish the norm equivalence between the target system and the original system with control.  From the transformation \eqref{eq:trans-simplify1} and \eqref{eq:trans-simplify2}, we get
\begin{align}
 & \|z \|_{L^2}^2
\le 4(1+\bar{K}^2)\|x \|_{L^2}^2 + 4\bar{c}^2(1+B_{\tau} ^2\bar{K}^2)\|u \|_{L^2}^2,\label{ieq:z_le_x_u}
\end{align}
which gives
\begin{equation}
  \|z \|_{L^2}^2+\|u \|_{L^2}^2\le 
k_1\left(\|x \|_{L^2}^2 +\|u \|_{L^2}^2\right),\label{ieq:z_u_le_x_u}
\end{equation}
with $k_1=4\max\{(1+\bar K), ~\bar c^2 (1+B^2_\tau \bar K^2)+1\}$.
Since $V_1\le \|z\|_{L^2}^2\le \mathrm{e}^{b_1} V_1$ and $\frac{\mathrm{e}^{-b_2}}{B_\tau}V_2\le \|u\|_{L^2}^2\le \frac{1}{\underline B_\tau} V_2$, we have
\begin{equation}
    \beta_1 V \leq \|z\|_{L^2}^2+\|u\|_{L^2}^2 \leq \beta_2 V,
\end{equation}
where $\beta_1$ and $\beta_2$ are defined in \eqref{def_beta1_2}.

For inverse transformation in \eqref{eq:inverse_trans}, it is worth noting that Delta function in $F_2(s,q,r)$ doesn't influence the boundedness of the norm of $x$ because the  Dirac Delta function can be eliminated by integration. The proof of Thereom 2 in \cite{zhang2021compensation} establishes this result, namely, 
\begin{align}\label{eq:inverse_kernel2_wp}
 F_2(s,q,r)=\left\{ 
    \begin{aligned}
         &\Xi(s,q,r),&s+\tau(q)r\le1, \\
         & \Xi(s,q,r) + \sum^{\infty}_{n=1}F_{22}^n,&s+\tau(q)r\le q,\\
          &0,&s+\tau(q)r>1,
    \end{aligned}
    \right.
\end{align}
where 
\begin{align}
    \Xi(s,q,r) &= -\delta(s-q+\tau(q)r)c(q)\tau(q),\\
 \left|  \sum^{\infty}_{n=1}F_{22}^n\right| &\le \bar{c}B_{\tau}\mathrm{e}^{\bar{f}B_{\tau}}.
\end{align}

Thus, we obtain from the inverse transformation that
\begin{equation} \label{ieq:inversemap}
\|x\|_{L^2}^2\le 4\left(1+\bar {F}_1^2\right)\|z\|_{L^2}^2 +4\bar c ^2(1+B_{\tau}^2\mathrm{e}^{2\bar{f}B_{\tau}})\|u\|_{L^2}^2,
\end{equation} 
where $\bar {F}_1=\sup_{(s,q)\in \mathcal{T}_1} |F_1(s,q)|$, and thus
\begin{align}\label{ieq:inv_trans1}    
\|x\|_{L^2}^2+ \|u\|_{L^2}^2\le  k_2 \left( \|z\|_{L^2}^2+ \|u\|_{L^2}^2\right).
\end{align}
Therefore, the norm equivalence between $\|x\|_{L^2}^2+\|u\|_{L^2}^2$ and $V(t)$ can be expressed as
\begin{equation}\label{ieq:normequiv}
    \frac{\beta_1}{k_1}V\leq \|x\|_{L^2}^2+\|u\|_{L^2}^2 \le k_2\beta_2 V.
\end{equation}
Combining \eqref{ieq:V_le_V0} and \eqref{ieq:normequiv}, we finally arrive at \eqref{ieq:stability_estimate}.
 
Given any  $B_x,~B_u>0$, to ensure that  $\|x(t)\|_{L^2}^2+\|u(t)\|_{L^2}^2$ does not exceed their bounds  as $t\rightarrow \infty$, the following inequality must hold 
\begin{equation}
    \lim_{t\rightarrow \infty} \|x(t)\|_{L^2}^2+\|u(t)\|_{L^2}^2 \le k_2 \beta_2 A \epsilon^2 \le B_x^2 + B_u^2, 
\end{equation}
which leads to \eqref{eq:epsilon*}.
From the stability estimate  \eqref{ieq:stability_estimate}, we know that the decaying term depends on the initial conditions and reach its maximum  at $t = 0$. To ensure that the estimate remains within the prescribed bounds, the following condition must be satisfied:  
\begin{equation}
    \frac{k_1}{\beta_1} k_2 \beta_2 (\|x_0\|_{L^2}^2 + \|u_0\|_{L^2}^2) + k_2 \beta_2 A \epsilon^2 \leq B_x^2 + B_u^2,
\end{equation}
which yields
\begin{equation}
    \|x_0\|_{L^2}^2 + \|u_0\|_{L^2}^2 \leq \frac{\beta_1}{k_1} \left( \frac{B_x^2 + B_u^2}{k_2 \beta_2} - A \epsilon^2 \right)=B_0.
\end{equation}
Hence, the theorem is proved.    
\end{proof}
It is noteworthy that selecting larger bounds for $x$ and $u$ and reducing the approximation error $\epsilon$ can expand the range of initial conditions $B_0$ for the semiglobal stability of the system. 

\section{Numerical Results}\label{Numerical Simulations} 
We employ a single DeepONet to approximate the controller \eqref{eq:U-case1} and \eqref{eq:U-case2} with two branches and one of branch involving two types of kernel gains, providing a unified neural-based controller for PDE systems.
The simulation code is available on \href{https://github.com/jackyhum/NeuralOperatorFeedbackWithSpatiallyVaryingStateDelay}{GitHub}.

Since the data generation method in \cite{bhan2024neural} cannot be applied due to correlations between $x$ and $u$, we instead numerically solve \eqref{eq:main-x1}–\eqref{eq:Combine-bnd-u1} with controllers \eqref{eq:U-case1}–\eqref{eq:U-case2} via finite differences on $t\in[0,15]$, under various initial conditions $x_0(s)$ and delay functions $\tau(s)$.

The initial conditions and delay profiles are sampled from Chebyshev-type functions \cite{curry2007parameter} as follows, 
\begin{align}
    x_0(s)\sim&~ \mathcal{A}_1 \cos(\Gamma_1 \cos^{-1}(s - \varkappa)),&\\
    \tau (s) \sim& ~3 + \mathcal{A}_2 \cos(\Gamma_2 \cos^{-1}(s)),&\text{~for~} \tau \in \mathcal{D}_1, \label{eq:tauD1data} \\
    \tau (s) \sim&~ \mathcal{A}_3 \mathrm{e}^{\Gamma_3 s},&\text{~for~} \tau \in \mathcal{D}_2,\label{eq:tauD2data}
\end{align} 
where the coefficients are drawn from uniform distributions: $\mathcal{A}_1 \sim U[0.5, 8]$, $\Gamma_1 \sim U[0, 8]$, $\varkappa \sim U[0, 0.5]$, $\mathcal{A}_2 \sim U[-1, 1]$, $\Gamma_2 \sim U[0,8]$, $\mathcal{A}_3 \sim U[0.4, 0.8]$ and $\Gamma_3 \sim U[0.8, 2.4]$. 
Other parameters in the PDE plant are fixed as follows: $c(s) = 20(1 - s)$, and $f(s, q) = 5 \cos(2\pi q) + 5 \sin(2\pi s)$. The initial condition for $u$ is set to zero.

For $\tau \in \mathcal{D}_1$, we use a temporal step of $0.025$, producing $9.6 \times 10^5$ instances of $(\tau, x, u)$ from 1600 different pairs   $(x_0(s),\tau(s))$. For $\tau \in \mathcal{D}_2$, a finer step of $0.005$ yields another $9.6 \times 10^5$ samples from 320 different different pairs  $(x_0(s),\tau(s))$. In total, the dataset comprises $1.92 \times 10^6$ instances of $\tau(\cdot,t_i),x(\cdot,t_i),~u(\cdot,\cdot,t_i)$ for all $s,r\in[0,1]$.



\begin{figure}[hbt]
    \centering
    \includegraphics[width=0.48\textwidth]{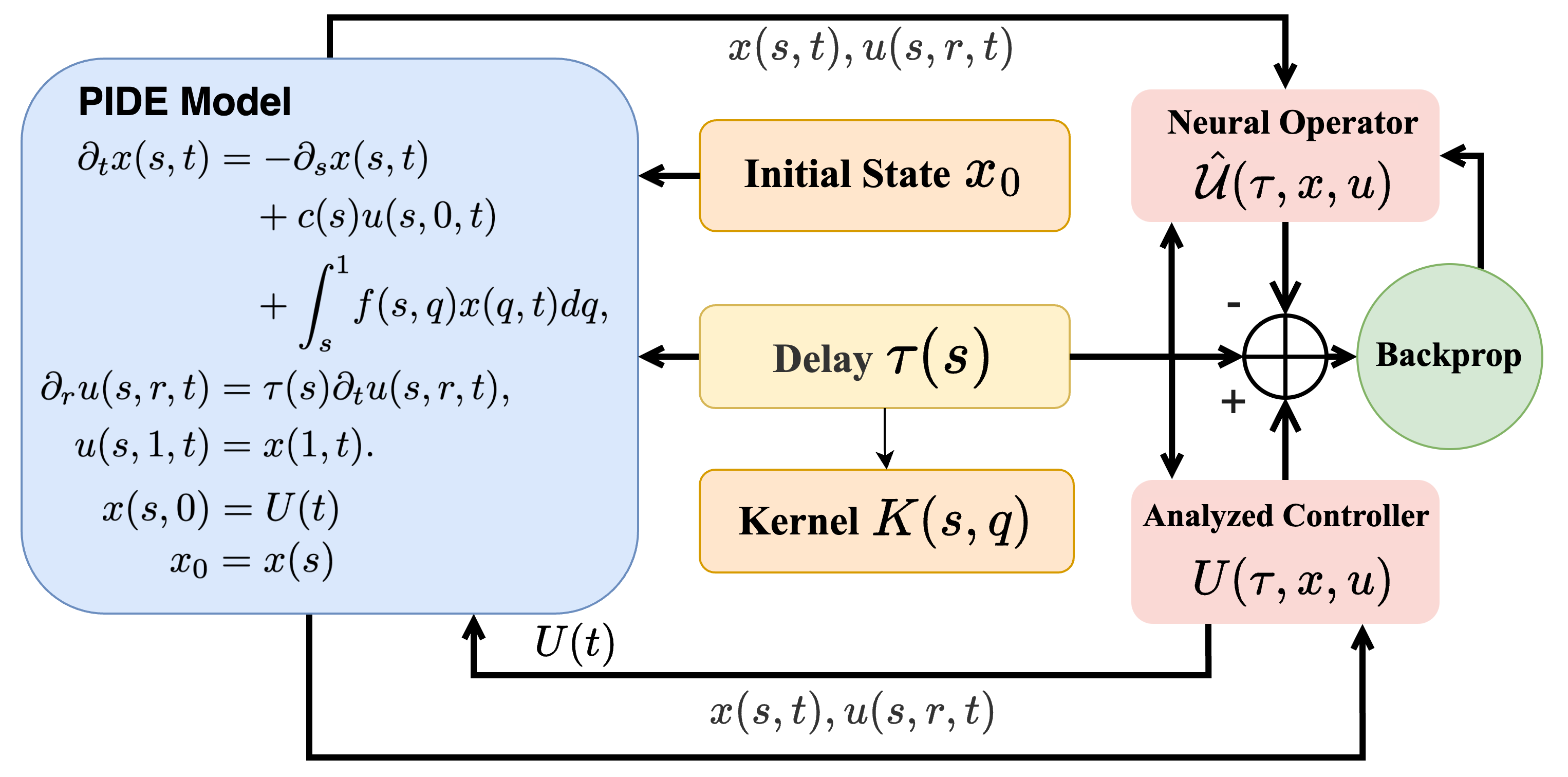}
    \caption{The neural operator training framework for the delay compensated controller.}
    \label{fig:generate_dataset}
\end{figure}

In this paper, we adopt the DeepONet architecture proposed in \cite{lu2021learning}, comprising a branch network and a trunk network. The branch network includes two convolutional layers (with kernel size $5 \times 5$ and stride $2$), followed by a fully connected layer of size $1152 \times 256$. The trunk network consists of two fully connected layers, whose input dimensions are determined by the spatial discretization of $(s, r)$ over the domain $[0,1]^2$.

We discretize the spatial domain for each training instance $(\tau,x,u,U)$ with step size $0.05$ on $[0,1]$, yielding $21$ grid points for $\tau(s)$ and $x(s)$, and $21 \times 21$ grid points for $u(s, r)$ at each time step. To align the domain of $\tau(s)$ and $x(s)$ with that of $u(s, r)$, their 1D representations are expanded to $[0,1]^2$, forming a $3 \times 21 \times 21$ tensor as input to the branch network.
In the trunk network, $21 \times 21$ grid are reshaped into a $441 \times 2$ array and processed through two fully connected layers, producing intermediate and final outputs of size $441 \times 128$ and $441 \times 256$, respectively.

We employ the smooth $L_1$ loss function introduced in \cite{ren2016faster}. Training the network, which contains approximately 3 million parameters, takes around 3 hours on an NVIDIA RTX 4090 GPU and achieves a final approximation loss of $5.89 \times 10^{-4}$ after $250$ epochs.

\begin{figure}[hbt]
  \centering
  \includegraphics[width=0.24\textwidth]{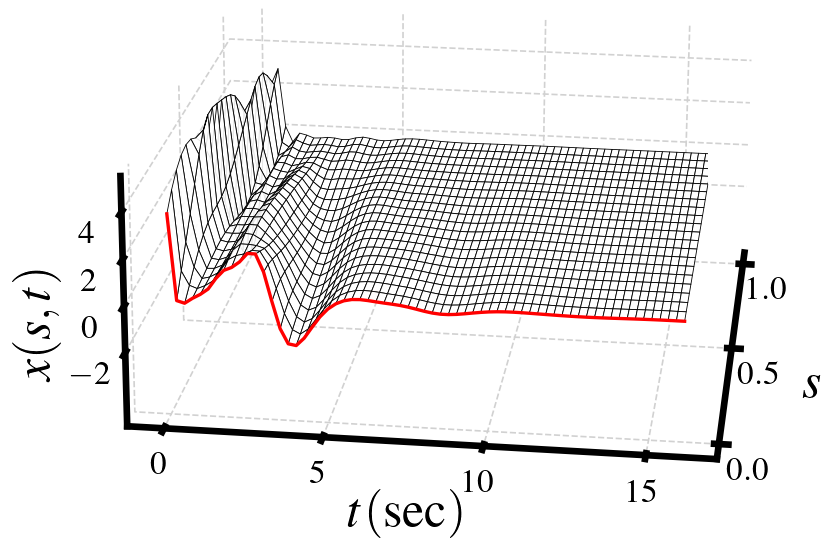} 
  \includegraphics[width=0.24\textwidth]{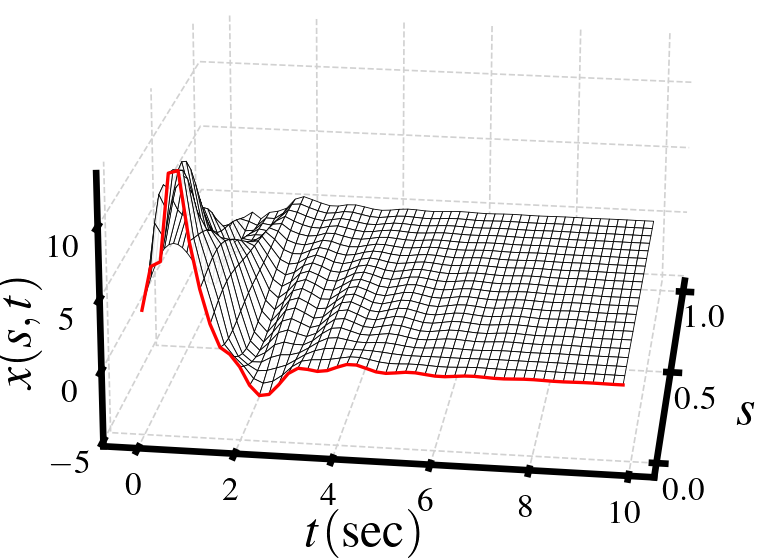}
    \\
\includegraphics[width=0.24\textwidth]{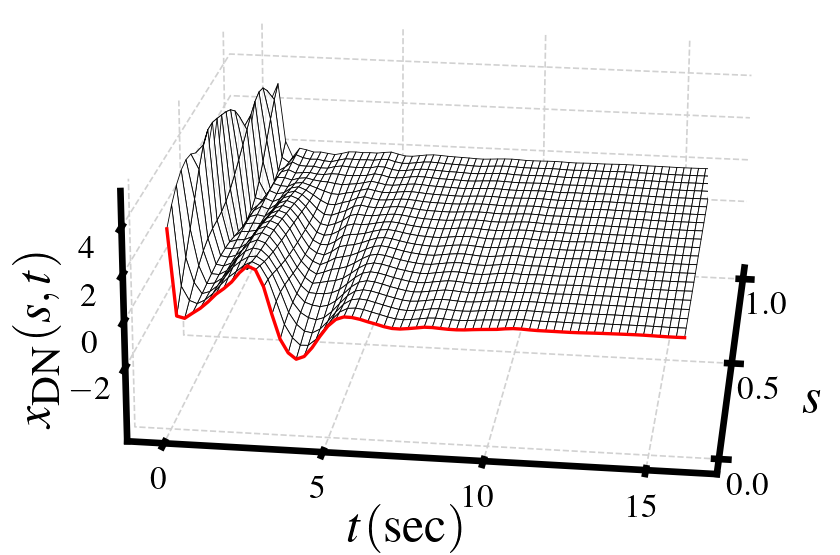}
\includegraphics[width=0.24\textwidth]{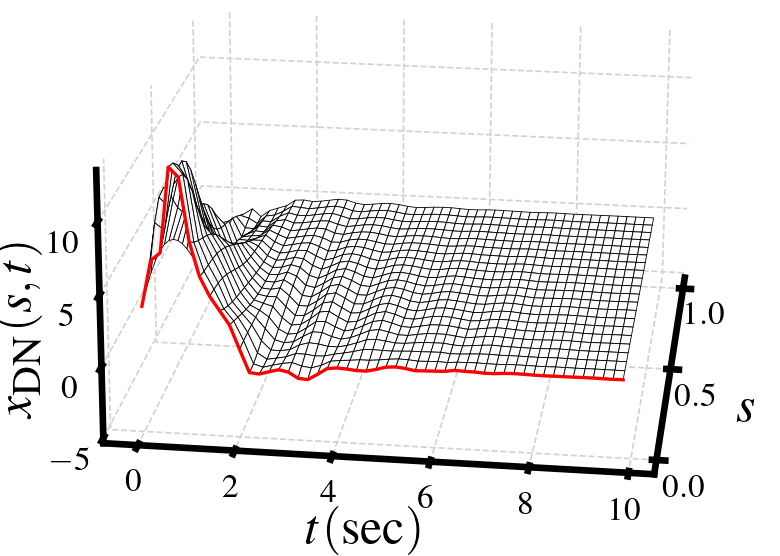}
   \\ 
\includegraphics[width=0.24\textwidth]{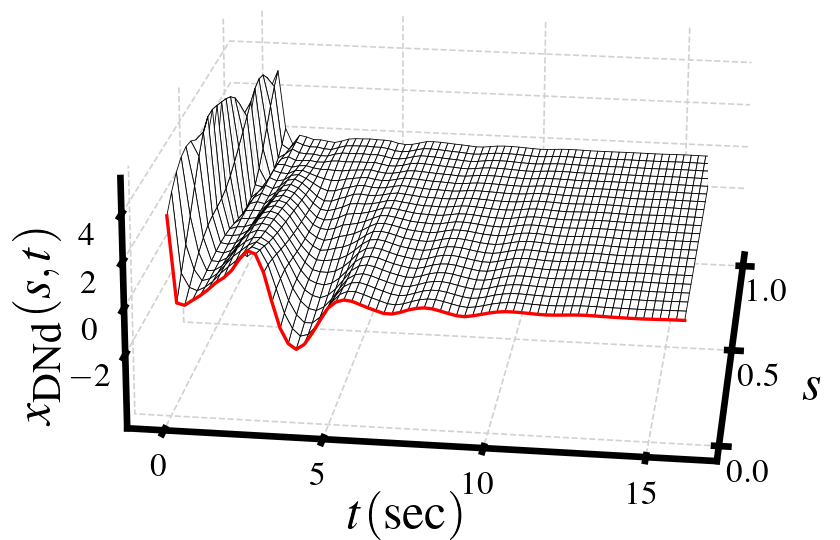}
\includegraphics[width=0.24\textwidth]{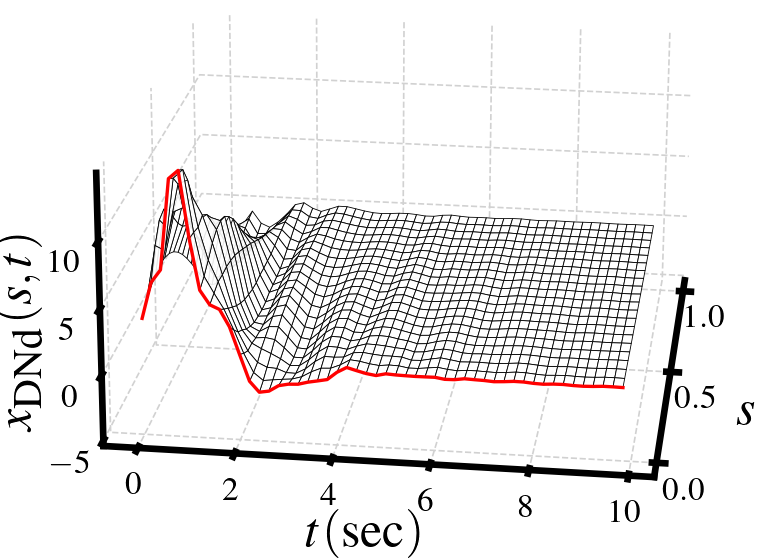}
  \caption{Closed-loop state $x(s,t)$ with initial condition $x_0=5\cos{(4\cos^{-1}(s-0.2))}$. Left: $\tau(s)=3+0.5 \cos{(5\cos^{-1}(s)}) \in \mathcal{D}_1$. Right: $\tau(s)=0.5e^{-1.6s} \in \mathcal{D}_2$. Top to bottom: states with the backstepping controller, NO-based controller, and NO-based controller for the delay with measurement noise (Gaussian noise $\mathcal{N}(0, \sigma^2)$).}
  \label{fig:dynamics_compare}
\end{figure}
\begin{figure}[hbt]
	\centering
    \begin{tabular}{cc}
        \includegraphics[width=0.23\textwidth]{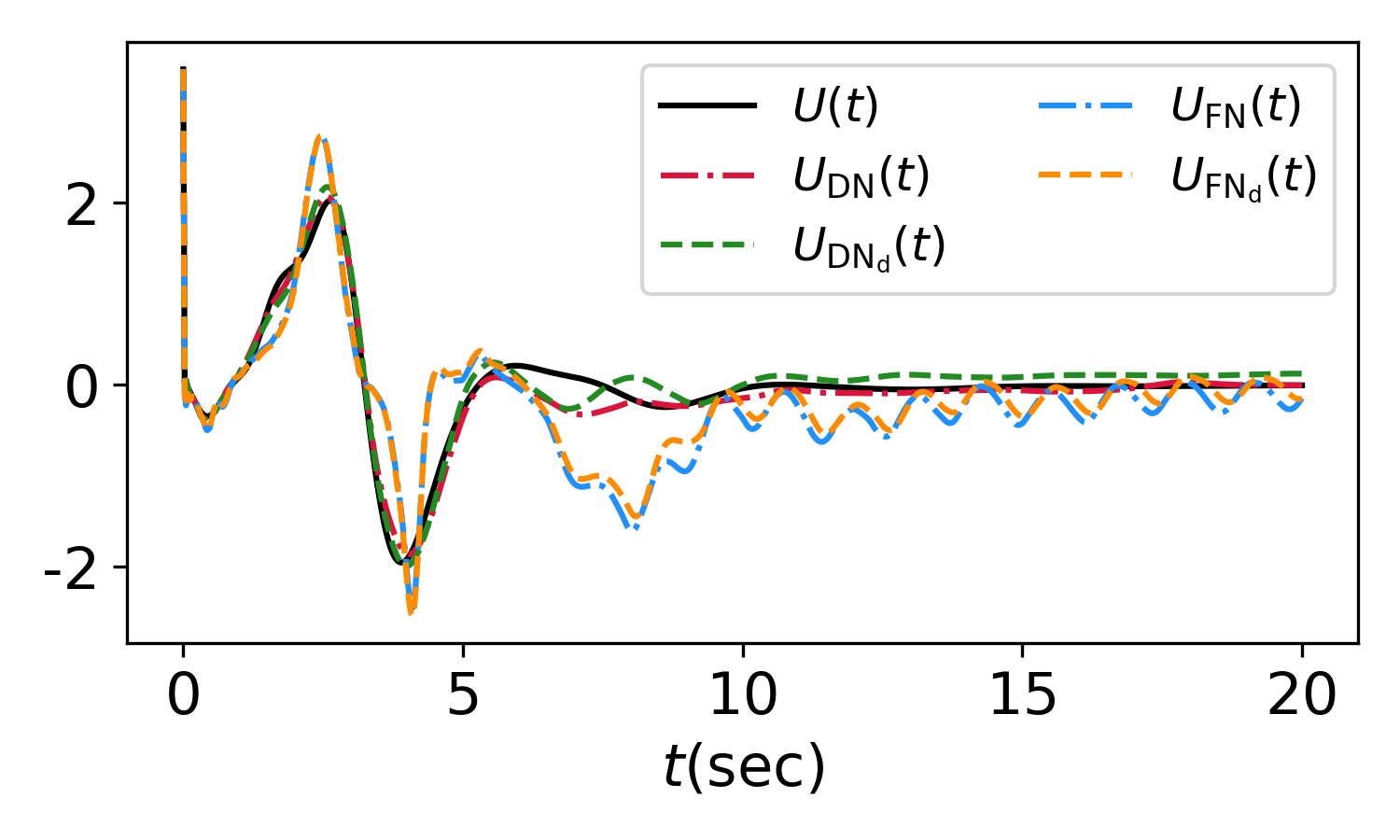}&\includegraphics[width=0.23\textwidth]{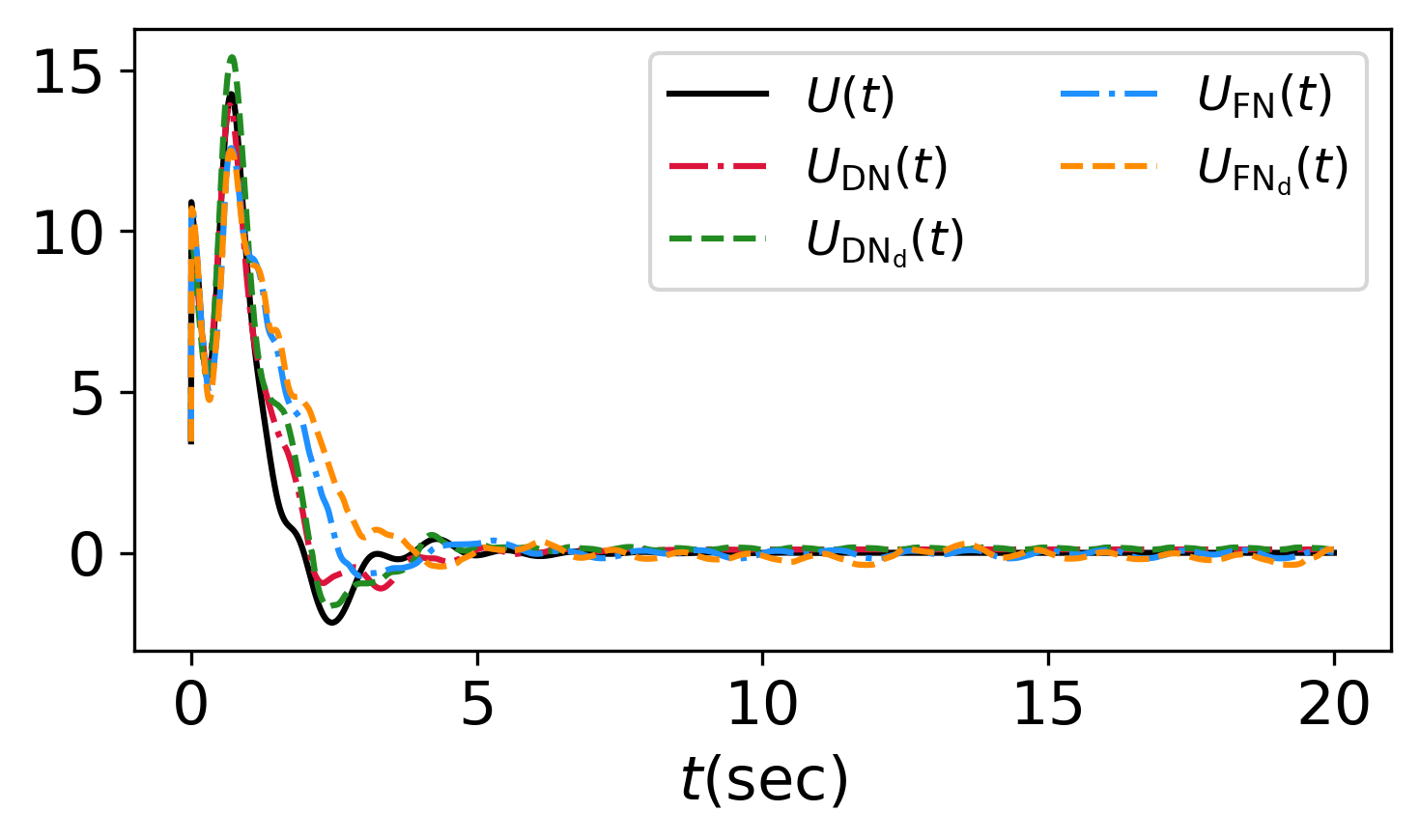}
  \end{tabular} 
   \begin{tabular}{cc}
        \includegraphics[width=0.23\textwidth]{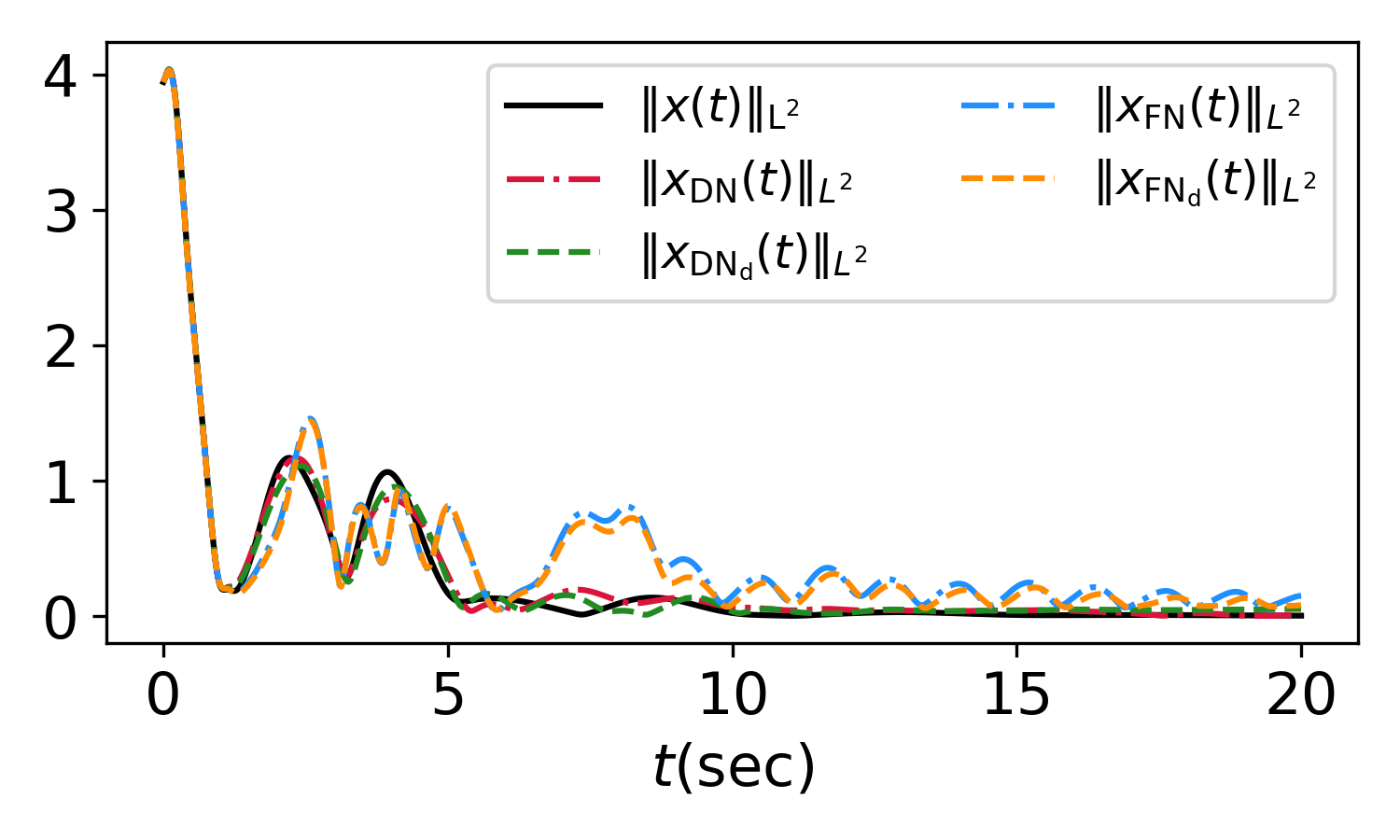}&\includegraphics[width=0.23\textwidth]{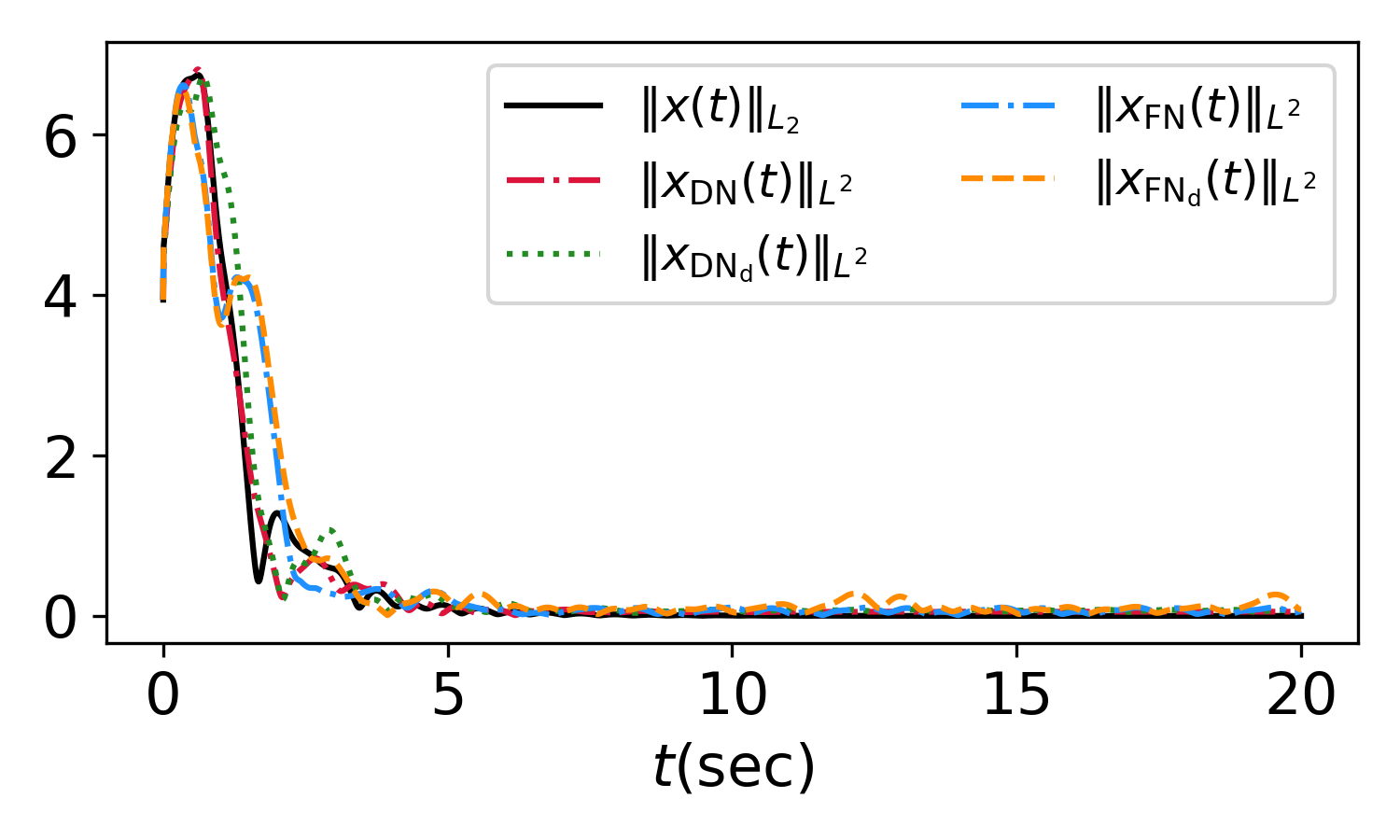}
  \end{tabular}
     \begin{tabular}{cc}
        \includegraphics[width=0.23\textwidth]{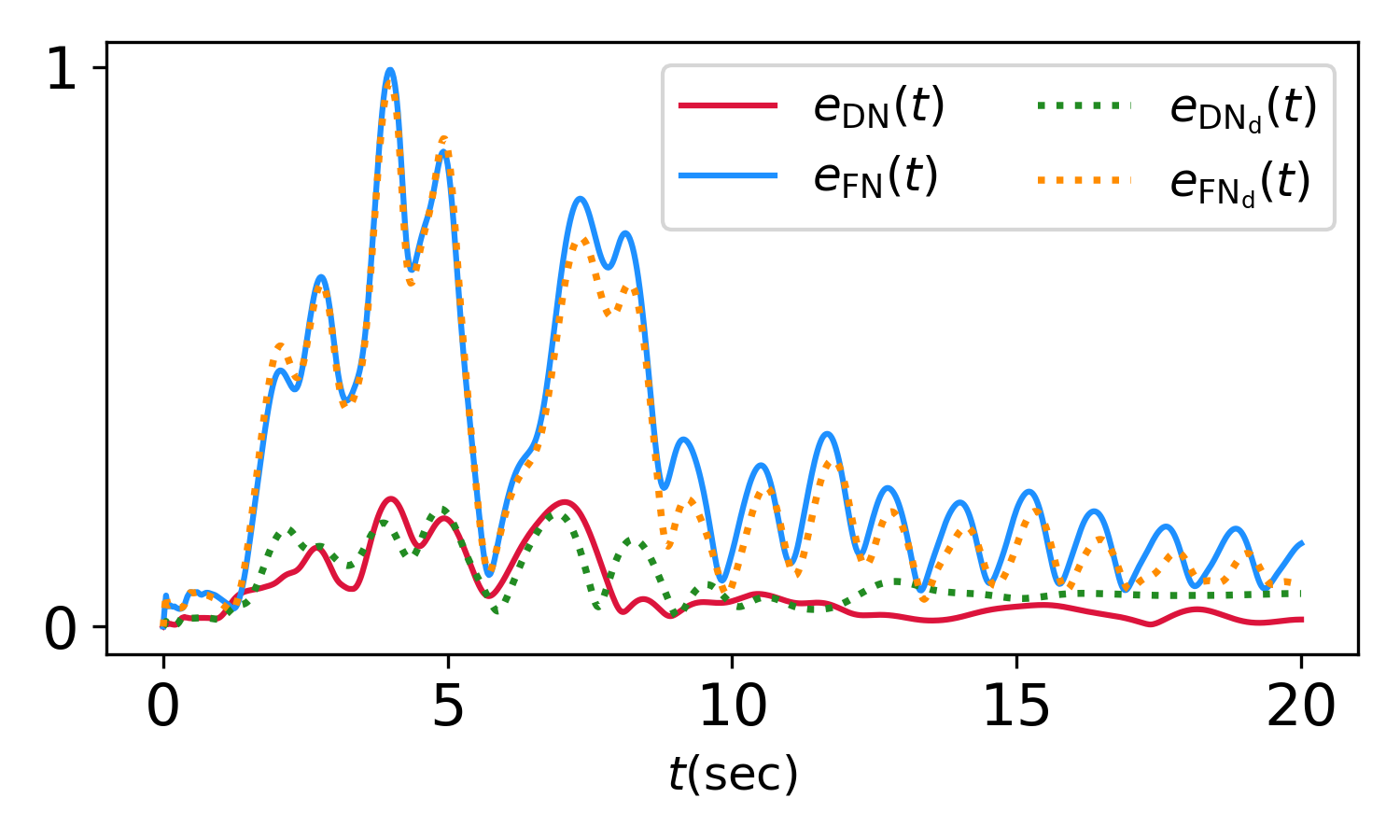}&\includegraphics[width=0.23\textwidth]{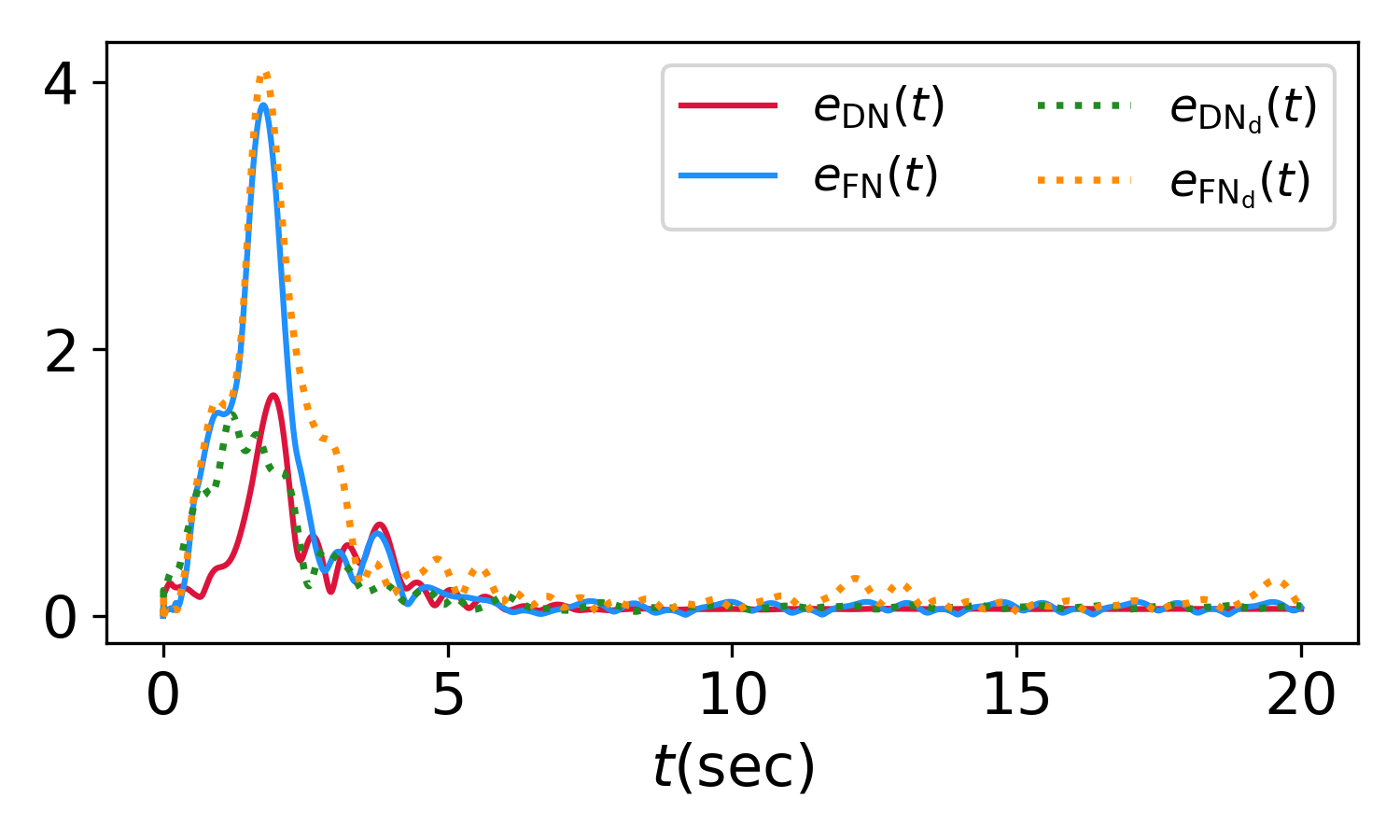}
  \end{tabular}
	\caption{From top to bottom: control input $U(t)$, state $x(s,t)$, and state error between neural operator controllers and the backstepping controller. Results are shown for DeepONet (`DN'), FNO (`FN'), DeepONet with noisy delay (`DN$_\mathrm{d}$'), and FNO with noisy delay (`FN$_\mathrm{d}$'). Left and right panels correspond to $\tau \in \mathcal{D}_1$ and $\tau \in \mathcal{D}_2$, respectively.}\label{fig:L2nornm}
\end{figure} 

 Fig.~\ref{fig:dynamics_compare} illustrates the closed-loop system states under both the backstepping controller and the DeepONet-based controller for two representative delay types: $\tau \in \mathcal{D}_1$ and $\tau \in \mathcal{D}_2$. When applied to a system with noisy delay (Gaussian noise with standard deviation $\sigma = 0.05$), the DeepONet-based controller demonstrates robust performance.

To benchmark DeepONet against other neural operator architectures, we also apply the Fourier Neural Operator (FNO)\cite{li2021fourier} to learn the backstepping controller. As shown in Fig.~\ref{fig:L2nornm}, the DeepONet controller exhibits lower overshoot, faster convergence, and consistently smaller state errors under both deterministic and noisy delays. The state error in the simulation is defined as 
\begin{align}\label{eq:state_err}
e(t) = \left( \sum_{i=1}^{n} \Delta s\left| x(s_i, t) - x_{\mathrm{NO}}(s_i, t) \right|^2   \right)^{\frac{1}{2}},
\end{align}
where $n = 21$ is the number of discretized spatial points, $s_i$ is the $i^{\mathrm{th}}$ position, and $\Delta s = \frac{1}{n-1}$ is the spatial step size.

\begin{table}[htbp]
    \centering
     \caption{Comparison between the numerical  controller solved by the finite difference method and NO-based controller over $t\in [0,16]$ for various spatial step sizes.}   \label{tab:compare}
    \begin{tabular}{lccc}
        \toprule
       \makecell{Spatial Step} &  \makecell{Average Numerical\\Solver Time\\Spent (sec)} & \makecell{Average Neural\\Operator Time \\Spent (sec)} &  Speedups \\
        \midrule
        $0.08$ & $7.2$ & $0.64$ & $11.3\times$\\
        $0.05$ & $22.37$ & $0.66$ & $33.9\times$ \\
        $0.025$ & $104.11$ & $0.73$&$142.6\times$ \\
        \bottomrule
    \end{tabular}
\end{table} 
 
Table \ref{tab:compare} represents a comparisons of computation time between the backstepping controller and the trained NO-based controller, averaged over 30 independent runs.
 It is evident that the DeepONet achieves at least an $11\times$ speedup compared to the  backstepping controller, which involves solving the backstepping kernel equations and performing numerical integration of the product of the kernel gain and the states. 

\section{Conclusion}\label{Conclusions}
This paper extends the Neural Operator (NO)-based control framework to handle spatially varying delays, simplifying traditional methods that required separate training for each kernel function and control branch. We train a single NO to approximate the control law, covering both branches and eliminating the need for kernel function selection.
We prove the Lipschitz continuity of the control operator with respect to the delay and states, and thus establish the semi-global practical stability of the closed-loop system. Simulations show the NO-based controller effectively compensates for spatially varying delays, with a computational speedup of $11\times$ over numerical methods, and robustness to small noisy delays.
Future work could extend the NO-based controller to delay-adaptive systems.

 \section*{References}
\vspace{-0.4cm}
\bibliographystyle{IEEEtranS}
\bibliography{IEEEabrv,learning_controller}  

\appendices
\section{Backstepping Kernel}
\renewcommand{\theequation}{A.\arabic{equation}}
\setcounter{equation}{0}
   An illustrative example of $g(q)$ defined in \eqref{eq:g} is shown in Fig. \ref{fig:curve}.
 \begin{figure}[hbt]
    \centering 
    \includegraphics[width=0.4\textwidth]{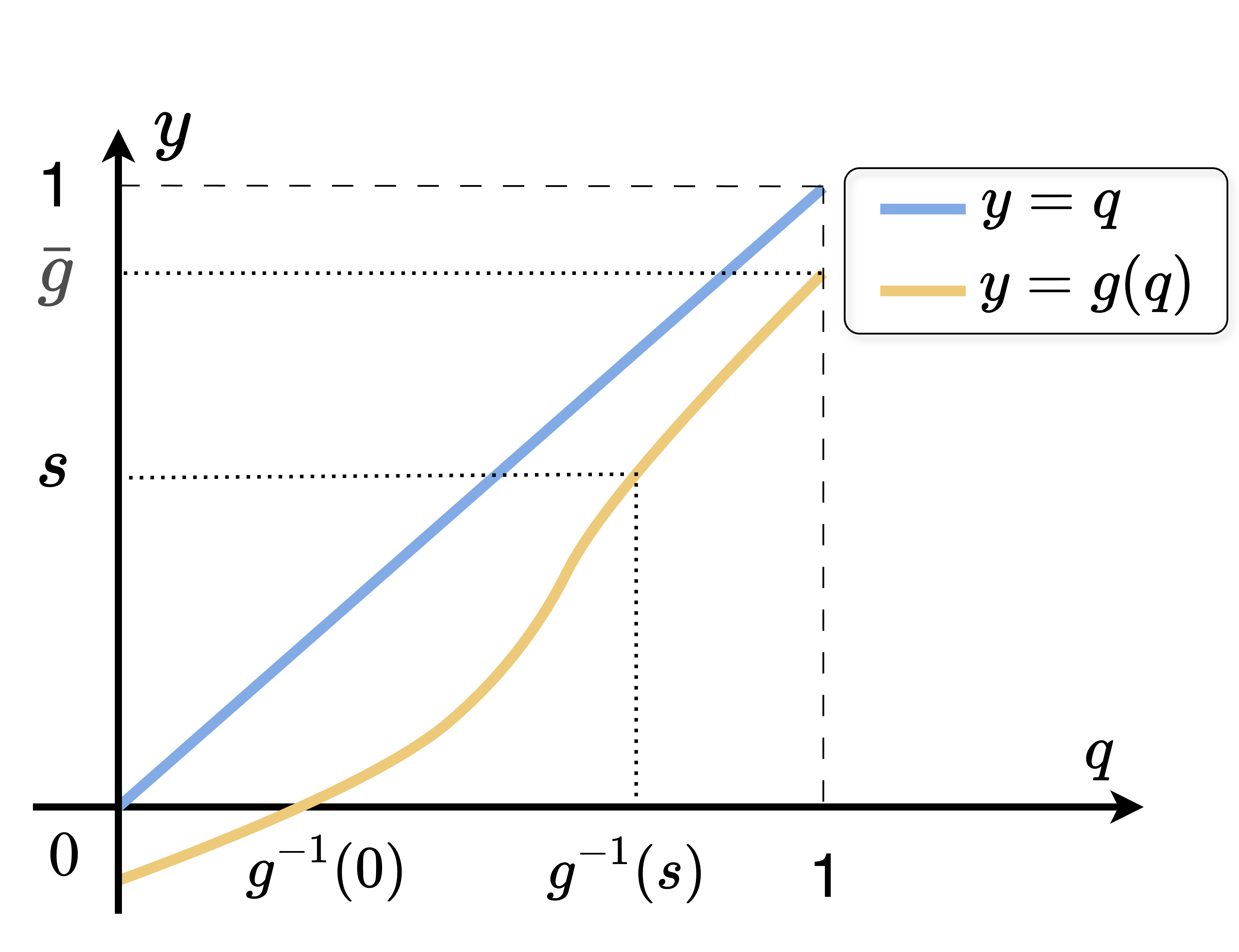}
   \caption{An illustrative example of $g(q)$ with its supremum denoted by $\bar g$.}
   \label{fig:curve}
  \end{figure}

\begin{theorem}\label{bound-kernel} \normalfont
    (Boundedness of kernel function) For $(c,f,\tau)$ $\in C^1[0,1]\times C^1(\mathcal{T}_1)\times \mathcal{D}$, the kernel function defined in \eqref{eq:kernel_K}-\eqref{eq:K_case2} has a unique solution $K\in C^0(\mathcal{T}_1)$, with bounded by
    \begin{align}
		\left|K(s,q)\right|\le  \bar{K}:=\frac{1}{w}  W_0 \mathrm{e}^{w(\bar c +\bar f)},  
	\end{align}
	where $w=\max\left\{1,\,{1}/{\underline g'}\right\}$ and $W_0=\left(\bar{c}/{\left(1-\bar{\tau}'\right)}+\bar{f}\right)$.  
\end{theorem}

The proof is presented in the \cite{ZHANG2024105964}.  

Note that 
If $\tau(q) \in \mathcal{D}_1$, the kernel $K(s, q)$ is determined by \eqref{eq:K_case1}, with a numerical example illustrated in Fig. \ref{fig:kernl_easy}. Conversely, if $\tau(q) \in \mathcal{D}_2$, $K(s, q)$ is governed by both \eqref{eq:K_case1} and \eqref{eq:K_case2}, with a numerical example shown in Fig. \ref{fig:kernl_complex}.  
   \begin{figure}[hbt]
    \centering 
    \includegraphics[width=0.38\textwidth]{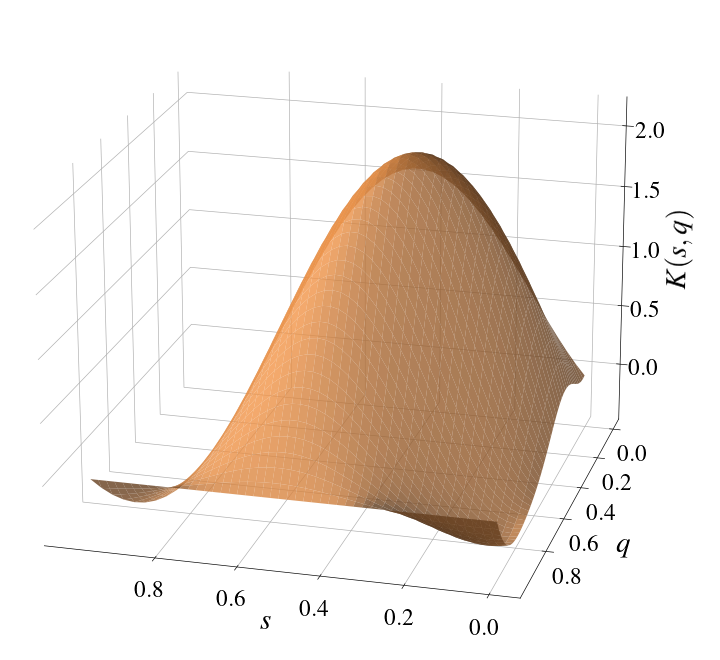}
   \caption{Kernel function $K(s,q)$ for $\tau(q) \in \mathcal{D}_1$ with $c=20(1-s)$, $f(s,q)=5\cos(2\pi s)+5\sin(2\pi q)$ and $\tau(s)=4-0.5\mathrm{e}^{s}$.}
   \label{fig:kernl_easy}
  \end{figure}

  \begin{figure}[hbt]
    \centering 
    \includegraphics[width=0.38  \textwidth]{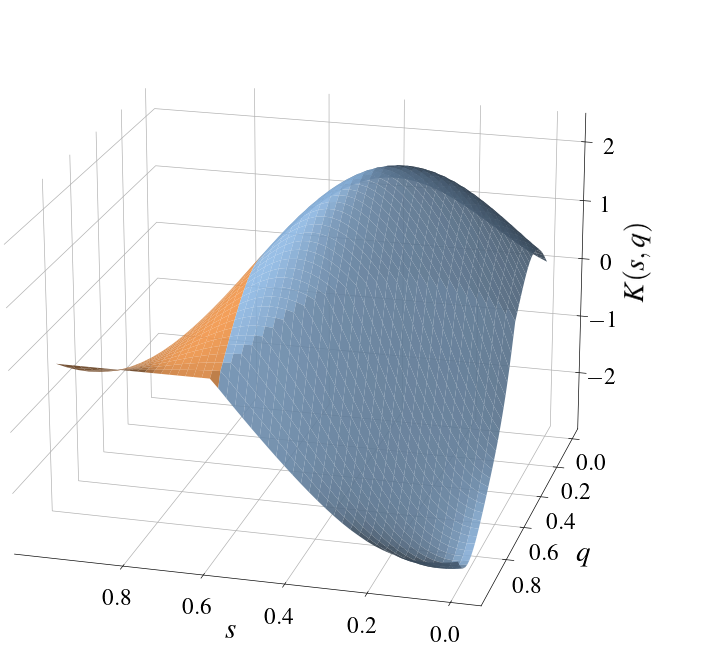}
   \caption{Kernel function $K(s,q)$  for $\tau(q) \in \mathcal{D}_2$ with $c(s)=20(1-s)$, $f(s,q)=5\cos(2\pi s)+5\sin(2\pi q)$ and $\tau(s)=2\mathrm{e}^{-2s}$. The blue  is governed by \eqref{eq:K_case1}, while the orange surface is governed by  \eqref{eq:K_case2}. }
   \label{fig:kernl_complex}
  \end{figure}


\section{Proofs of Lemma \ref{Lemma:inverse_g_0} and Lemma \ref{lemma:lip_k_tau_0}}\label{proof:lemma1_2}
\renewcommand{\theequation}{B.\arabic{equation}}
\setcounter{equation}{0}
 We begin with the proof of Lemma \ref{Lemma:inverse_g_0} as outline below.

\begin{proof}


For $\tau_1,\tau_2\in \mathcal{D}_2$ and any $\sigma \in \text{Ran}(g)$, we evaluate 
\begin{align*}
    g_1(g_2^{-1}(\sigma))=g_2(g_2^{-1}(\sigma))+\left(g_1(g_2^{-1}(\sigma))-g_2(g_2^{-1}(\sigma))\right),
\end{align*}
so we get
\begin{align}
    &g_1(g_1^{-1}(\sigma))-g_1(g_2^{-1}(\sigma))=\sigma-g_1(g_2^{-1}(\sigma))
    \nonumber 
    \\&=g_2(g_2^{-1}(\sigma))-g_1(g_2^{-1}(\sigma)).
\end{align}
Apply mean value theorem for $g_1$, 
\begin{align}
    &g_1(g_1^{-1}(\sigma))-g_1(g_2^{-1}(\sigma))=g'_1(\zeta)(g_1^{-1}(\sigma)-g_2^{-1}(\sigma)).
\end{align}
for a $\zeta \in (0,1)$. Recall $g'(s)\ge \underline g'$, which yields,
\begin{align}
|g_1^{-1}(\sigma)-g_2^{-1}(\sigma)| 
&\le 1/\underline g' |g_1(g_2^{-1}(\sigma))-g_2(g_2^{-1}(\sigma))|
\nonumber 
\\ & \le  1/\underline g' \|g_1-g_2\|_\infty.
\end{align}
  
Therefore,
\begin{align}
\|g_1^{-1}-g_2^{-1}\|_\infty 
&\le \frac{\|g_1-g_2\|_\infty}{\underline g'} \nonumber \\
& \le 1/\underline g'\|\tau_1-\tau_2\|_\infty,
\end{align}
which proves \eqref{ieq-g-inverse_same_section}.

Now we consider the second inequality. Let $\tau_1 \in \mathcal{D}_1$, $\tau_2 \in \mathcal{D}_2$ and $g_2(q^*)=0$, which gives 
$g_1(1)\le 0$, $g_2(1) > 0$ and $q^*=\tau_2(q^*)$.
Consider 
\begin{align}
    g_2(1)= g_2(1)-g_2(q^*)=g'_2(\zeta)(1-q^*)\ge \underline g'(1-q^*),
\end{align}
for $\zeta \in(0,1)$, which gives
\begin{align}
    \underline g' (1-q^*)\le g_2(1) \le |g_1(1)-g_2(1)|=|\tau_2(1)-\tau_1(1)|.
\end{align}
Recalling $q^*=g_2^{-1}(0)$, we obtain
\begin{align}
    \|1-g_2^{-1}(0)\|_\infty \le \|\tau_2(1)-\tau_1(1)\|_\infty /\underline g',
\end{align}
 which proves \eqref{ieq-g-inverse_diff_section}. 

\end{proof}

Before proving the Lemma \ref{lemma:lip_k_tau_0}, we introduce the following lemmas.

\begin{definition}\label{def:F}
 The operator $\mathcal{F}: \mathcal{D} \mapsto C^1[0,1]$ with
 \begin{align}\label{ope-U}
 \Xi_2(\sigma) =\mathcal{F}( \tau)(\sigma)
 \end{align} 
is defined by the expressions \eqref{eq:Xi2} with $\sigma=s-q+1$.
\end{definition}

\begin{lemma}\label{Lemma:F_Lips}\normalfont
    Let $\tau_1,\tau_2 \in \mathcal{D}_2$. The operator $\mathcal{F}$ ope-defined in \eqref{ope-U} is Lipschitz continuous, satisfying 
    \begin{align} \label{ieq:F_Lipschitz}
        \|\mathcal{F}(\tau_1)-\mathcal{F}(\tau_2)\|_{\infty}\le L_F\|\tau_1-\tau_2\|_{\infty},
    \end{align}
    where $L_F>0$.
\end{lemma}
\begin{proof}
Let $g_i(s) = s - \tau_i(s)$ and $h_i := g_i^{-1}$ for $i = 1,2$. Then
\begin{align*}
|\mathcal{F}(\tau_1)(\sigma) - \mathcal{F}(\tau_2)(\sigma)|
= \left| \frac{c(h_1)}{g_1'(h_1)} - \frac{c(h_2)}{g_2'(h_2)} \right| 
\le I + II,
\end{align*}
where
\begin{align*}
I := \left| \frac{c(h_1) - c(h_2)}{g_1'(h_1)} \right|,\quad 
II := \left| c(h_2)\left( \frac{1}{g_1'(h_1)} - \frac{1}{g_2'(h_2)} \right) \right|.
\end{align*}

Since $g_i'(s)  > \underline g' > 0$, both denominators are bounded below. By the Lipschitz continuity of $c$, $|I| \le \frac{L_c}{\underline g'} |h_1 - h_2|$. For $II$, we write
\begin{align}
\left| \frac{1}{g_1'(h_1)} - \frac{1}{g_2'(h_2)} \right| 
&= \left| \frac{g_2'(h_2) - g_1'(h_1)}{g_1'(h_1)g_2'(h_2)} \right| 
\nonumber 
\\&
\le \frac{1}{\underline g'^2} |\tau_1'(h_1) - \tau_2'(h_2)|.
\end{align}
Then, using the Lipschitz continuity of $\tau_i'$, we get
\begin{align}
    |\tau_1'(h_1) - \tau_2'(h_2)| 
\le& |\tau_1'(h_1) - \tau_1'(h_2)| + |\tau_1'(h_2) - \tau_2'(h_2)| 
\nonumber 
\\\le& L_{\tau'} |h_1 - h_2| + \|\tau_1' - \tau_2'\|_\infty
\nonumber 
\\\le& (L_{\tau'}/\underline g'+1)  \|\tau_1' - \tau_2'\|_\infty,
\end{align}
where we use the inequality \eqref{ieq-g-inverse_same_section}. Finally, we arrive at 
\eqref{ieq:F_Lipschitz} with $L_F=\frac{1}{\underline g'^3}(\bar c L_{\tau'}+\bar c \underline g'+L_c\underline g'^2)$.
\end{proof}

\begin{lemma}\normalfont \label{Lemma:funtion_K_Lips}
      For $(s,q)\in \mathcal{T}_1$, $K_1(s,q)$ and $K_2(s,q)$ defined in \eqref{eq:K_case1} and \eqref{eq:K_case2}, respectively, are Lipschitz continuous with respect to $s$, that is, for  $0\le s_1,s_2\le q$, we have
        \begin{align}
        |K_1(s_1,q) - K_1(s_2,q)| \le& L_1 | s_1 - s_2|,\\
   \left|K_2(s_1,q) - K_2(s_2,q)\right| \le& L_{2} | s_1 - s_2|,
        \end{align}
        with Lipschitz constants   $L_1, L_2 >0$.
    \end{lemma}
\begin{proof}\normalfont
 We first consider
  \begin{align}
	&|K_1(s_1,q) - K_1(s_2,q)| \nonumber 
 \\=& ~|\Psi_1(K_1)(s_1,q)-\Psi_1(K_1)(s_2,q)
 -\Xi_1(s_1,q)+\Xi_1(s_2,q)|\nonumber 
 \\\le &~ L_1 |s_1-s_2| ,\label{ieq:K_1s1-s2}
	\end{align}
 where $L_1=3\bar f\bar K+(1+\bar K) L_f$ and $\Psi_1$ and $\Xi_1$ are defined in \eqref{eq:Psi1} and \eqref{eq:Xi1}, respectively. 
 
For more complex kernel function $K_2(s,q)$, and we denote  $\delta_s(K_2):=K_2(s_1,q) - K_2(s_2,q)$.
 Then, rewrite   
 \begin{align}
&\delta_s(K_2)=\delta_s(\Phi_0)+\delta_s(\Phi_1(K_2)),
\end{align}
where
\begin{align*}
\Phi_0&(s,q)=\Psi_1(K)(s,q)+\Psi_{21}(K_1)(s,q)-\Xi_1(s,q)\nonumber 
\\&-\Xi_2(s,q)
+\int_{\psi(s)}^1 c(q)K_2(s_1+1-q+\tau(p),p)dp,\\
\Phi_1&(K_2)(\sigma)=\int_{\psi(s_2)}^{1} c(p)K_2(\sigma+\tau(p),p)dp,\nonumber 
\end{align*}
with $\sigma=s+1-q$, $\sigma_i=s_i+1-q$ for $i=1,2$ and $\Psi_{21}$,  $\Xi_2$ and $\psi$ defined in \eqref{eq:Psi21}
 \eqref{eq:Xi2} and \eqref{eq:psi}. 
 Consider the iteration
 \begin{align}
&\delta_s(K^{n+1}_2)=\delta_s(\Phi_0)+\delta_s(\Phi_1(K^n_2)),
\end{align}
and let
 \begin{align}
     \Delta^n\delta_s(K_2)=&\delta_s(K_2^{n+1})-\delta_s(K_2^{n})
      \\
     \Delta^0\delta_s(K_2)=&\delta_s(\Phi_0). \label{eq:initial_Delta_K_s}
 \end{align}
 which gives 
 \begin{align}
     \Delta^n\delta_s(K_2)=&\Phi_1(\Delta^{n-1}\delta_s(K_2)), \label{eq:Delta_K_s}
     \end{align}
 If the series $\Delta^n\delta_s(K_2)$ converges, we have
 \begin{align} 
     \delta_s(K_2)=\sum_{n=0}^\infty \Delta^n\delta_s(K_2).
 \end{align}
We begin with initial value \eqref{eq:initial_Delta_K_s}. In the first term of $\Phi_0$, only $f$ and the integration limits depend on $s$, yielding 
\begin{align}\label{ieq:Phi_01}
    |\delta_s\Psi_1(K)|\le \bar K(L_f+3\bar f)|s_1-s_2|.
\end{align}
Combining \eqref{ieq:K_1s1-s2}, the second term of $\Phi_0$ satisfies
\begin{align} \label{ieq:Phi_02}
    |\delta_s\Psi_{21}(K_1)|\le&~\bar c \bar K|\psi(s_1)-\psi(s_2)| 
    +\bar c| \delta_s (K_1)| \\  
      & +\bar c \bar K (g^{-1}(s_1)-g^{-1}(s_1))
    \nonumber 
    \\\le & ~\bar c(2\bar K L_g +3\bar f\bar K+L_f(1+\bar K))|s_1-s_2|\nonumber.
\end{align}
The third and fourth terms in $\Phi_0$ satisfy
\begin{align}\label{ieq:Phi_034}
    &~| \delta_s(\Xi_1)+\delta_s(\Xi_2)|
    \nonumber \\\le
     &~L_f|s_1-s_2|+\frac{\bar c}{(1-\bar \tau')^2}|g'(g^{-1}(\sigma_2))-g'(g^{-1}(\sigma_1))|
       \nonumber \\ 
     &~+\frac{\bar g'}{(1-\bar \tau')^2}
    |c(g^{-1}(\sigma_1))-c(g^{-1}(\sigma_2))|
    \nonumber \\\le &~
    \left(L_f+\frac{L_g}{(1-\bar \tau')^2}(\bar c L_{\tau'}+L_c \bar g')\right)|s_1-s_2|.
\end{align}
Recalling $s$ only appears in the integration limits of the last term of $\Phi_0$,  and combining \eqref{ieq:Phi_01}-\eqref{ieq:Phi_034}, we derive
\begin{align}
    |\delta_s(\Phi_0)|\le L_0 |s_1-s_2|,
\end{align}
where 
\begin{align}
    L_0=&3\bar c\bar K L_g+L_f(1+\bar K)(1+\bar c)+\bar  3\bar f\bar K(1+\bar c)\nonumber 
    \\&+\frac{L_g}{(1-\bar \tau')^2}(\bar c L_{\tau'}+L_c \bar g').
\end{align}
Second, we apply the successive approximation method to prove that  series \eqref{eq:Delta_K_s} converges.
Assume that 
\begin{align}\label{eq:induction_delta_sK}
    |\Delta^n\delta_s(K_2)|\le L_0|s_1-s_2|\frac{\bar c^n}{(\underline g' )^{n-1}n!}(q-s_2)^n
\end{align}
Substitute \eqref{eq:induction_delta_sK} into \eqref{eq:Delta_K_s}, we get
\begin{align}
    |\Delta^{n+1}\delta_s(K_2)|\le &~\frac{L_1\bar c^{n+1}|s_1-s_2|}{(\underline g')^{n-1}n!}\left|\int_{\psi(\sigma_2)}^1 \frac{(g(p)-\sigma_2)^n}{g'(p)}dg(p)
    \right| \nonumber 
    \\\le&~\frac{L_0\bar c^{n+1}|s_1-s_2|}{(\underline g')^{n}(n+1)!}   
    (g(1)-s_2-1+q)^{n+1}
    \nonumber
    \\\le &~\frac{L_0\bar c^{n+1}|s_1-s_2|}{(\underline g')^{n}(n+1)!}   
    (q-s_2)^{n+1},  
\end{align}
which proved the assumption \eqref{eq:induction_delta_sK}.
Therefore,
\begin{align}
    |\delta_s(K_2)|\leq L_0\underline g'\mathrm{e}^{\bar c/\underline g'} |s_1-s_2|,
\end{align}
which results in $L_2=L_0\underline g'\mathrm{e}^{\bar c/\underline g'}$.
The Lemma is proved. 
\end{proof}

Now, we present the proof of Lemma \ref{lemma:lip_k_tau_0} as outlined below. 
\begin{proof}
First, we rewrite \eqref{eq:Psi1} as 
\begin{align}
    \Psi_1(K,\tau)=\Psi_{11}(K_1,\tau)+\Psi_{12}(K_2,\tau),
\end{align}
where 
\begin{equation}
    \Psi_{11}(K_1,\tau)=\int_s^{\sigma}\int_\theta^{\theta-\tau(1)}K_1(\theta,r)f(r,\theta-s+q)dr d\theta,
    \end{equation}
    \begin{equation}
    \Psi_{12}(K_2,\tau)=\int_s^{\sigma}\int^{\theta-s+q}_{\theta-\tau(1)}K_2(\theta,r)f(r,\theta-s+q)dr d\theta,
\end{equation}
 Using the simplified notations $K_{21}=\mathcal{K}_2(\tau_1)(s,q)$ and $K_{22}=\mathcal{K}_2(\tau_2)(s,q)$, we have 
 \begin{align}
 K_{21}-K_{22}&= \Psi_{11}(K_1,\tau_1)-\Psi_{11}(K_1,\tau_2) +\Psi_{12}(\mathcal{K}_2(\tau_1),\tau_1)\nonumber \\
&-\Psi_{12}(\mathcal{K}_2(\tau_2),\tau_2)+\Psi_{21}(K_1,\tau_1)-\Psi_{21}(K_1,\tau_2)\nonumber 
\\& +\Psi_{22}(\mathcal{K}_2(\tau_1),\tau_1)-\Psi_{22}(\mathcal{K}_2(\tau_2),\tau_2)\nonumber 
\\&-\frac{c\left(g_1^{-1}(\sigma)\right)}{g_1'\left(g_1^{-1}(\sigma)\right)}+\frac{c\left(g_2^{-1}(\sigma)\right)}{g_2 '\left(g_2^{-1}(\sigma)\right)}.
\end{align}
We rewrite it as
\begin{align}
K_{21}&-K_{22}=\Phi_0(s,q)+\Phi_1(K_{21}-K_{22}),
\end{align}
where
\begin{align}\label{eq:Phi0}
    &\Phi_0=-\frac{c\left(g_1^{-1}(\sigma)\right)}{g_1'\left(g_1^{-1}(\sigma)\right)}+\frac{c\left(g_2^{-1}(\sigma)\right)}{g_2 '\left(g_2^{-1}(\sigma)\right)}
    \\& +\int_s^{\sigma}\int_{\theta-\tau_2(1)}^{\theta-\tau_1(1)}(K_1-\mathcal{K}_2(\tau_1))(\theta,r)f(r,\theta-s+q)drd\theta
    \nonumber \\&
    +\int_{g_1^{-1}(\sigma)}^{g_2^{-1}(\sigma)} cK_{1}(\sigma+\tau_1(p),p)dp\nonumber 
    \\&   +\int_{\kappa_2}^{\kappa_1} cK_{1}(\sigma+\tau_1(p),p)dp  
    \nonumber \\&
     +\int_{g_2^{-1}(\sigma)}^{\kappa_2} c(K_{1}(\sigma+\tau_1(p),p)-K_{1}(\sigma+\tau_2(p),p))dp
     \nonumber \\&
     +\int_{\kappa_1}^{\kappa_2} c\mathcal{K}_{2}(\tau_1)(\sigma+\tau_1(p),p)dp
       \nonumber \\&
     +\int^{1}_{\kappa_2} c(\mathcal{K}_{2}(\tau_1)(\sigma+\tau_1(p),p)-\mathcal{K}_{2}(\tau_1)(\sigma+\tau_2(p),p))dp, \nonumber
\end{align}
with
\begin{align*}
    \kappa_i(s,q,\tau_i)&=g^{-1}_i(\min\{(\bar g_i, \sigma+\tau_i(1)\}),~ i=1,2,  
\end{align*}
and 
\begin{align*}
    \Phi_1(\mathcal{K}_2(\tau))=&\int_s^{\sigma}\int_{\theta-\tau_2(1)}^{\theta-s+q} (\mathcal{K}_{2}(\tau) )(\theta,r) f(r,\theta-s+q)drd\theta\nonumber 
    \\&+\int^{1}_{\kappa_2} c\mathcal{K}_{2}(\tau)(\sigma+\tau_2(p),p )dp.
\end{align*}
Then,  for $n=0,1,2,\ldots$, we have
\begin{align*}
\Phi_1(\mathcal{K}^{n+1}_2(\tau))=&\int_s^{\sigma}\int_{\theta-\tau_2(1)}^{\theta-s+q} (\mathcal{K}^n_{2}(\tau) )(\theta,r) f(r,\theta-s+q)drd\theta\nonumber 
    \\&+\int^{1}_{\kappa_2} c\mathcal{K}^{n}_{2}(\tau)(\sigma+\tau_2(p),p )dp.
\end{align*}

 Let
 \begin{align}\label{eq:sum_K2}
     K_{21}-K_{22}=\sum_{n=0}^\infty \Delta^n(K_{21}-K_{22}),
 \end{align}
 where
 \begin{align*}
     \Delta^n(K_{21}-K_{22})&=(K^{n+1}_{21}-K^{n+1}_{22})-(K^{n}_{21}-K^n_{22}),\\
     \Delta^0(K_{21}-K_{22})&=\Phi_0,
 \end{align*}
 which yields
 \begin{align}
     \Delta^n(K_{21}-K_{22})&= \Phi_1(\Delta^{n-1}(K_{21}-K_{22})).
\end{align}
First, we consider
\begin{align}
    |\Phi_0 (s,q)|\le& L_F\|\tau_1-\tau_2\|_{\infty}+2\bar K\bar f \|\tau_1-\tau_2\|_{\infty}\nonumber 
    \\& +2\bar c\bar K|\kappa_1-\kappa_2|_{\infty}+\bar c (L_1+L_2) \|\tau_1-\tau_2\| _{\infty} \nonumber\\&+
    \frac{\bar c\bar K}{\underline g'} \|\tau_1-\tau_2\| _{\infty},
\end{align}
where we use   Lemma \ref{Lemma:F_Lips}
 and Lemma \ref{Lemma:funtion_K_Lips}.
 Since $|\kappa_1-\kappa_2|\le 2L_g\|\tau_1-\tau_2\|_{\infty} +\frac{1}{\underline g'}\|\tau_1-\tau_2\| _{\infty}$, we obtain
 \begin{align}
    |\Phi_0 &(s,q)|\le L_{\Phi_0}\|\tau_1-\tau_2\|_{\infty},
\end{align}
where 
\begin{equation}
    L_{\Phi_0}=L_F+2\bar K\bar f + \frac{3\bar c\bar K}{\underline g'}+\bar c(L_1+L_2)
    +4L_g\bar c\bar K.
\end{equation}
 Assume
\begin{align}
     |\Delta^n(K_{21}-K_{22})|\le& \left|\Phi_0\right|\frac{\max\{\bar f^n,\bar c^n\}(1-s)^n}{ n!},
 \end{align}  
 which gives
 \begin{align}
    & |\Delta^{n+1}(K_{21}-K_{22})| \nonumber \\\le&
     \int_s^{s+1-q}  \bar f |\Phi_0|\frac{\max\{\bar f^n,\bar c^n\}(1-\theta)^n}{ n!} d\theta\nonumber 
    \\&+\int^{1}_{s+1-q} \bar c|\Phi_0|\frac{\max\{\bar f^n,\bar c^n\}(1-p)^n}{ n!} dp,
    \nonumber
    \\ \le& 
     \int_s^{1}  \max\{\bar c,\bar f\} |\Phi_0|\frac{\max\{\bar f^n,\bar c^n\}(1-\theta)^n}{ n!} d\theta\nonumber
    \nonumber
    \\ \le&    |\Phi_0|\frac{\max\{\bar f^{n+1},\bar c^{n+1}\}(1-\theta)^{n+1}}{ {n+1}!} 
 \end{align}  
 where we use $\kappa_2(s,q,\bar g_2)>s+1-q$. 
 Therefore, we finally arrive at \eqref{ieq:Ktau1-Ktau2}, with the Lipschitz constant $L_{K}=L_{\Phi_0} \mathrm{e}^{\max\{\bar f,\bar c\}}$.
\end{proof}


\section{Stability of the closed-loop sytem in the $C^1$ norm under  backstepping control}\label{C1stability1}
\renewcommand{\theequation}{D.\arabic{equation}}
\setcounter{equation}{0}
 For  $f(s)\in L^p[0,1]$ and $g(s,q)\in L^p ([0,1]^2)$, where $p\in \mathbb{N}^+$ and $\mu\ne 0$, we define
 \begin{align}
     	\lVert f \rVert_{\mu, p}&=  \left(\int_{0}^{1}|\mathrm{e}^{p\mu s}f(s)|^{p}ds \right)^{\frac{1}{p}},\\
      \lVert g \rVert_{\mu,p}&=  \left(\int_{0}^{1}\int_{0}^{1}|\mathrm{e}^{ \mu ps}g(s,q) |^{p}dqds.\right)^{\frac{1}{p}}.
 \end{align}
For $z(s) \in C^1[0,1]$ and $u(s,r)$, we define
 \begin{align}
	\lVert z\rVert_{C}~=&\,\lVert z \rVert_\infty, \label{def:normC}\\
 	\lVert z\rVert_{C^1}=&\,\lVert z\rVert_C+\lVert   \partial_{s} z\rVert_C,\label{def:normC1}\\
  \lVert u\rVert_{C^1}=&\,\lVert u\rVert_C+\lVert   \partial_{r} u\rVert_C.\label{def:normC1_u}
 \end{align}
 {\setlength{\parskip}{5pt}
 \begin{prop} \normalfont
    Consider the closed-loop system \eqref{eq:main-x0}-\eqref{eq:x-before} under the control law \eqref{eq:U-case1} and  \eqref{eq:U-case2}. For initial conditions $(x_0, u_0)  \in C^1[0,1]\times C^1([0,1]^2)$ that are compatible with the boundary conditions, the system is exponentially stable in $C^1$ norm, specifically, there exists a positive constant $M_0$ and $\alpha>0$ such that for $t\ge 0$,
      \begin{align}\label{eq:U_lipshciz_contion0}
          \|x(t)\|_{C^1}+\|u(t)\|_{C^1}\leq 
          \mathrm{e}^{-\alpha t} M_0( \|x_0\|_{C^1}+\|u_0\|_{C^1}).
      \end{align}
\end{prop}
\begin{proof}
    We first consider the target system and define the following Lyapunov functions:  
\begin{align}
	V_{1p}(t)=&\int_0^1{\mathrm{e}^{-2bps}  z(s,t)^{2p}}ds\label{eq:V1p}\\
	          &+\int_0^1{\mathrm{e}^{-2bps}  \partial_{s}z(s,t)^{2p}}ds,\nonumber\\
	V_{2p}(t)=&\int_0^1{\int_0^1{\tau(s)\mathrm{e}^{2bpr} u(s,r,t)^{2p}}}drds\label{eq:V2p}\\
			  &+\int_0^1{\int_0^1{\tau(s)\mathrm{e}^{2bpr} \partial_{r}u(s,r,t)^{2p}}}drds,\nonumber\\
     	V_p(t)=&A_1 V_{1p}(t)+V_{2p}(t),\label{eq:Vp}
\end{align}
where  $b\in \mathbb{R}^+$ , $ p\in \mathbb{N}^+$ and $A_1 > 0$.
Moreover, from system \eqref{eq:main-tar-z}-\eqref{eq:bnd-tar-u}, we obtain 
\begin{align}
	\partial_{st} z(s,t)&=-\partial_{ss}z(s,t),\\
 \partial_{t}z(0,t)&=\partial_{s}z(0,t)=0,\\
	\tau(s)\partial_{rt}u(s,r,t)&=\partial_{rr}u(s,r,t),\\
	\partial_{r}u(s,1,t)&=-\tau(s)\partial_{s}z(1,t).
\end{align}
Then taking the time derivative of \eqref{eq:V1p} -\eqref{eq:Vp}, we obtain 
\begin{align}
	\dot V_{1p}(t)&=-\mathrm{e}^{-2b p}\left(z(1,t)^{2p}+\partial_{s}z(1,t)^{2p}\right)-2b pV_{1p}(t),\\
 \dot V_{2p}(t)&\le  a_1 \mathrm{e}^{2b p}\left(z(1,t)^{2p}+\partial_{s}z(1,t)^{2p}\right)-\frac{2b p}{\bar{\tau}}V_{2p}(t),
\end{align}
where $a_1=\max\{1,\int_0^1\tau(s)^{2p}ds\}$.
Let $A_1=a_1 \mathrm{e}^{4bp}$ and $a_2=\min\{1,1/\bar{\tau}\}$, there exists 
\begin{align}
    	\dot{V}_p(t)&\le -2bp \left( A_1 V_{1p}+\frac{1}{\bar{\tau}}V_{2p}\right) \leq -2 a_2bp  V_{p}.\label{eq:c1_derivative_v1}
\end{align}
From \eqref{eq:c1_derivative_v1}, we then get $V_p^{(1/2p)}(t)\le \mathrm{e}^{-a_2bt}V_p^{(1/2p)}(0)$ and combine with the Power Mean inequality
\begin{align}
    \frac{A_1^{\frac{1}{2p}}V_{1p}^{\frac{1}{2p}}+V_{2p}^{\frac{1}{2p}}}{2} \leq \left(\frac{A_1V_{1p}+V_{2p}}{2}\right)^{\frac{1}{2p}},
\end{align}
  to arrive at 
\begin{align}
    A_1^{\frac{1}{2p}}V_{1p}^{\frac{1}{2p}}+V_{2p}^{\frac{1}{2p}}\le 3\mathrm{e}^{-a_2bt}( A_1^{\frac{1}{2p}}V_{1p}^{\frac{1}{2p}}(0)+V_{2p}^{\frac{1}{2p}}(0)), \label{eq:c1_1}
\end{align}
With the definitions \eqref{eq:V1p} and \eqref{eq:V2p},  we are confident that \eqref{eq:c1_1} can be rewritten as norm as follow
\begin{align}\label{ieq:weight_norm_ieq}
\chi_p(t)\le& 6\mathrm{e}^{-a_2bt}\chi_p(0),\\
    \chi_p(t)=& A_1^{\frac{1}{2p}}\left(\|z(t)\|_{-b,2p}+\|\partial_{s}z(t)\|_{-b,2p}\right)\\&+\left(\|u(t)\|_{b,2p}+\|\partial_{r}u(t)\|_{b,2p}\right).\notag 
\end{align}
 Take the limit of \eqref{ieq:weight_norm_ieq} as $p\rightarrow +\infty$, with $C^1$ norm defined in \eqref{def:normC1} and \eqref{def:normC1_u}, we obtain 
\begin{equation}\label{eq:inequality1}
\lVert z(t) \rVert_{C^1}+ \lVert u(t) \rVert_{C^1} \le W \mathrm{e}^{-a_2bt}(\lVert z(0)\rVert_{C^1}+ \lVert u(0) \rVert_{C^1}).
\end{equation}
with $W>0$, where we also used the fact  
\begin{align}
    \mathrm{e}^{-b}\|z(t)\|_C\le & \|z(t)\|_{-b,C}\le \|z(t)\|_{C}\\
    \|u(t)\|_C \le &\|u(t)\|_{b, C}\le \mathrm{e}^b\|u(t)\|_C.
\end{align}
Therefore we prove that the target system  is stable in $C^1$ norm.

We now proceed to prove the stability of original system \eqref{eq:main-x0}-\eqref{eq:x-before}, using the inverse transformation \eqref{eq:inverse_trans}, whose kernel functions $F_1$ and $F_2$ are well-posedness stated in Theorem 2 in paper \cite{zhang2021compensation, ZHANG2024105964}. More specifically, 
\begin{align}\label{eq:inverse_kernel2_wp}
 F_2=\left\{ 
    \begin{aligned}
         &\Xi(s,q,r),&q<s+\tau(q)r\le1, \\
         & \Xi(s,q,r) + Q(s,q,r) ,&s+\tau(q)r\le q,\\
          &0,&s+\tau(q)r>1,
    \end{aligned}
    \right.
\end{align}
where 
\begin{align}
    \Xi(s,q,r) &= -\delta(s-q+\tau(q)r)c(q)\tau(q),\\
 Q(s,q,r)&=\sum^{\infty}_{n=1}F_{22}^n(s,q,r). 
\end{align}
From paper \cite{zhang2021compensation}, we know functions  $F_1(s,q)$, $\partial_s F_1$, $Q(s,q,r)$ and $Q_s(s,q,r)$ are bounded with bounds $\|F_1(s,q)\|\le \bar {F}_1 $, $\|\partial_s F_1(s,q)\|\le \bar {F}_{1s} $, $|F_{22}^n(s,q,r)|\le \frac{\bar c\bar f^n\bar \tau^{n+1}}{n!}r^n$, which gives $\| Q (s,q,r)\|\le \bar {Q}:= \bar c\bar \tau \mathrm{e}^{\bar f\bar \tau }$ and $\|\partial_s Q (s,q,r)\|\le \bar {Q}_s:=2\bar c \bar \tau \mathrm{e}^{\bar f\bar \tau}  $.
Based on these bounds, we derive  
\begin{align}
    \|x\|_C \le & ~(1+\bar F_1)\|z\|_C+(\bar c  +\bar Q)\|u\|_C, \\
    \|\partial_{s}x\|_C \le &~ \|\partial_{s}z\|_C+(\bar F_1+\bar F_{1s}\|)\|z\|_C\nonumber\\
    &~+(\bar c+\bar Q_s) \|u\|_C + \frac{\bar c}{\underline \tau} \| \partial_{r} u\|_C .
\end{align}
Combining this with inequality \eqref{eq:inequality1}, we obtain 
\begin{align}
&~\|x\|_{C^1}+\|u\|_{C^1}\nonumber
	\nonumber\\\le&~(1+2\bar F_1 +\bar F_{1s})\|z\|_C+\|\partial_{s}z\|_{C}\\
	\nonumber&~+( 2\bar c+\bar Q +\bar Q_s    )\|u\|_C +\frac{\bar c}{\underline \tau} \|\partial_{r} u\|_{C}\\
	\nonumber\le&~M(\|z\|_{C^1}+\|u\|_{C^1})\\
	\le&~MW\mathrm{e}^{-a_2bpt}(\lVert z(0)\rVert_{C^1}+ \lVert u(0) \rVert_{C^1}), \notag
\end{align}
where  
\begin{equation}
    M=\max\{  1+2\bar F_1 +\bar F_{1s},\, 2\bar c+\bar Q +\bar Q_s   , \, {\bar c}/{\underline \tau}\}. 
\end{equation}
Therefore, we reach inequality \eqref{eq:U_lipshciz_contion0} by letting $M_0=MW$ and $\alpha=a_2bp$.
This completes the proof. 
\end{proof}
}
\section{Figures Illustrating the Closed-Loop System without Delay-Compensation}
\begin{figure}[!t]
    \centering 
    \begin{tabular}{cc}
          \includegraphics[trim={2pt 0pt 1.5pt 0pt}, clip,width=0.25\textwidth]{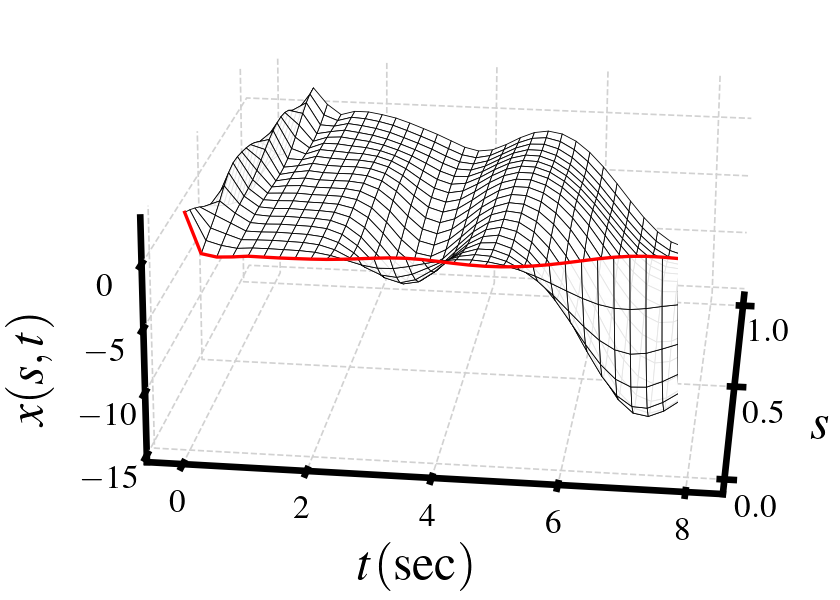}
          \includegraphics[trim={5pt 0pt  0pt 0pt}, clip, width=0.25\textwidth]{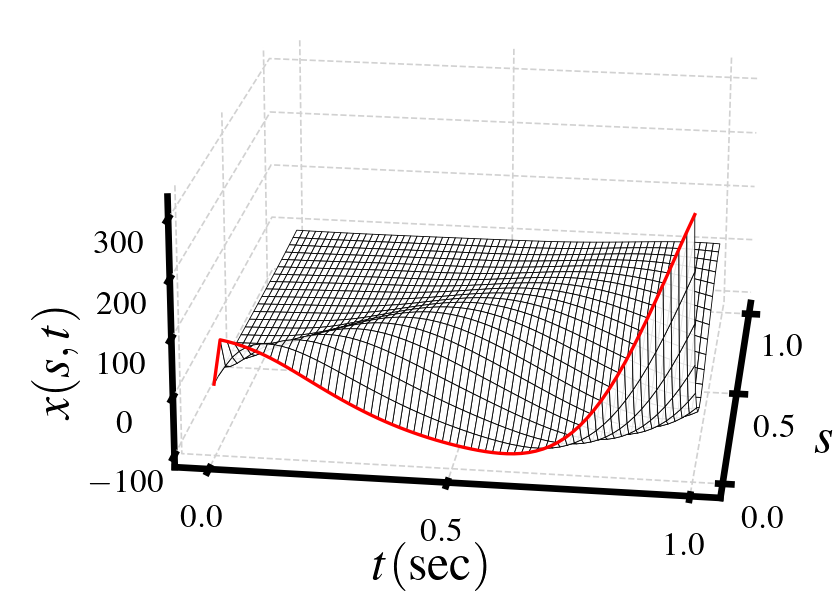}
    \end{tabular}
     \caption{Left: Dynamics of the closed-loop system without delay-compensation with $\tau(s)=3+0.5\cos{(5\cos^{-1}(s)}) \in \mathcal{D}_1$. Right: Dynamics for $\tau(s)=0.5e^{-1.6s} \in \mathcal{D}_2$. Both cases starting from initial condition $x_0=5\cos{(4\cos^{-1}(s-0.2))}$.}
     \label{fig:uncompensated}
  \end{figure}  
Fig. \ref{fig:uncompensated} shows that the closed-loop system dynamics without delay compensation fail to converge for both $\tau(s) \in \mathcal{D}_1$  and  $\tau(s) \in \mathcal{D}_2$.
\end{document}